\documentclass[a4paper,10pt]{article}
\usepackage{a4wide}
\usepackage{amsfonts}
\usepackage{amssymb}
\usepackage{amsmath}
\usepackage{enumerate}
\usepackage{epsfig}
\usepackage{cases}
\usepackage{color}
\usepackage{psfrag}
\usepackage[numbers,sort&compress]{natbib}

\def\qed{}

\newcommand{\ms}{\mathrm{ms}} 
\newcommand{\mV}{\mathrm{mV}} 
\newcommand{\su}{\mathrm{su}} 

\newcommand{\apriori}{\textit{a~priori}}
\newcommand{\defacto}{\textit{de~facto}}
\newcommand{\eg}{e.g.}
\newcommand{\etal}{\textit{et~al.}}
\newcommand{\ie}{i.e.}
\renewcommand{\qed}{\rule{1ex}{1ex}}
\newcommand{\QED}{\hfill\qed}
\newcommand{\viceversa}{vice versa}

\renewcommand{\@}{\partial}
\newcommand{\const}{\mathrm{const}}
\renewcommand{\d}{\mathrm{d}}
\newcommand{\Df}[2]{\frac{\d #1}{\d #2}}
\newcommand{\DDf}[2]{\frac{\d^2 #1}{\d #2^2}}
\newcommand{\df}[2]{\frac{\partial #1}{\partial #2}}
\newcommand{\ddf}[2]{\frac{\partial^2 #1}{\partial #2^2}}
\newcommand{\emb}[1]{\underline{{#1}}}
\newcommand{\eq}[1]{(\ref{eq:#1})}
\newcommand{\eqlabel}[1]{\label{eq:#1}}
\newcommand{\Fig}[1]{Figure~\ref{#1}}
\newcommand{\fig}[1]{figure~\ref{#1}}
\newcommand{\figs}[1]{figures~\ref{#1}}
\newcommand{\figref}[1]{\ref{#1}}
\newcommand{\hash}{\#}
\newcommand{\mx}[1]{\mathbf{#1}}
\renewcommand{\O}[1]{\mathcal{O}\left(#1\right)}
\renewcommand{\o}[1]{o\left(#1\right)}
\newcommand{\Qstat}[1]{#1_{\infty}}
\newcommand{\Real}{\mathbb{R}}

\newcommand{\dblfigure}[3]{
  \begin{figure*}[t]
    \centerline{\resizebox{\textwidth}{!}{\includegraphics{#1}}}
    \caption[]{#2}
    \label{#3}
  \end{figure*}
}
\newcommand{\sglfigure}[3]{
  \begin{figure}[t]
    \centerline{\includegraphics{#1}}
    \caption[]{#2}
    \label{#3}
  \end{figure}
}
\newtheorem{prop}{Proposition}
\newtheorem{con}{Conjecture}
\newenvironment{proof}{\textbf{Proof}\ }{\QED}

\newcounter{assm} 
\newenvironment{assm}{\begingroup%
  \begin{enumerate}%
  \setcounter{enumi}{\value{assm}}%
}{%
  \setcounter{assm}{\value{enumi}}%
  \end{enumerate}\endgroup%
}
\newcommand{\+}[3]{\def#1{{#2}}}
\+{\APD}{\mathrm{APD}}{Action potential duration}
\+{\A}{A}{ZFK rhs amplitude}
\+{\BCLmin}{\BCL_{\min}}{the minimal temporal period for which there are propagating waves}
\+{\BCL}{P}{the temporal period}
\+{\BesselI}{I}{the modified Bessel function of the 1st kind}
\+{\BesselK}{K}{the modified Bessel function of the 2nd kind}
\+{\Ca}{\mathrm{[Ca]}}{Ca++ ions concentration}
\+{\Cfast}{\C_{\textrm{fast}}}{faster solution for $\C$}
\+{\Ch}{\C^+}{max $\C$}
\+{\Clwr}{\underline{\C}}{lower bound for $\C$}
\+{\Cl}{\C^-}{min $\C$}
\+{\Cmax}{\C_\text{max}}{min $\C$}
\+{\Cmin}{\C_\text{min}}{max $\C$}
\+{\Cslow}{\C_{\textrm{slow}}}{slow solution for for $\C$}
\+{\Cupr}{\overline{\C}}{upper bound for $\C$}
\+{\C}{C}{non-dimensional wave speed}
\+{\DI}{\mathrm{DI}}{diastolic interval}
\+{\D}{D}{diffusion constant}
\+{\ENa}{\E_{\mathrm{Na}}}{Na reversal potential}
\+{\Ealpha}{\E_\alpha}{Prefront voltage}
\+{\East}{E_{\ast}}{location of max of the time-independent current in the caricature model}
\+{\Edagger}{E_{\dagger}}{location of min of the time-independent current in the caricature model}
\+{\Edel}{\beta}{inhom term coeff in slow piece $\E$ soln}
\+{\Edmp}{\alpha}{damping rate of $\E$ in slow pieces}
\+{\Edrei}{\overset{3}{\E}}{third case of the exact solution of caricature model}
\+{\Eeins}{\overset{1}{\E}}{first case of the exact solution of caricature model}
\+{\Efastalphamax}{{\Efastalpha}_{\max}}{max value of $\Efastalpha$ needed for front}
\+{\Efastalphamin}{{\Efastalpha}_{\min}}{min value of $\Efastalpha$ needed for front}
\+{\Efastalpha}{\Efast_{\alpha}}{prefront value of $\Efast$}
\+{\Efastomega}{\Efast_{\omega}}{postfront value of $\Efast$}
\+{\Efast}{V}{$\E$ value in the fast-time system}
\+{\Eh}{\E_h}{h-gate switch voltage}
\+{\EinfEdmp}{\gamma}{qstat val of $\E$ in a slow piece, times $\Edmp$}
\+{\Ejot}{\overset{\jp}{\E}}{$\jp$-piece of the CN slow solution for $\E$}
\+{\Em}{\E_m}{m-gate switch voltage}
\+{\Eomega}{\E_\omega}{Postfront voltage}
\+{\Eone}{E_1}{lower zero of the time-independent current in the caricature model}
\+{\Ethree}{E_3}{upper zero of the time-independent current in the	caricature model}
\+{\Etwo}{E_2}{middle zero of the time-independent current in the caricature model}
\+{\Ev}{\E_0}{pinning value of $\E$}
\+{\Ezwei}{\overset{2}{\E}}{second case of the exact solution of caricature model}
\+{\E}{E}{transmembrane voltage}
\+{\FF}{F}{IBessel's log log derivative}
\+{\Fone}{F_h}{constant value of $\fone$ in the caricature model}
\+{\Ftwo}{F_{n}}{constant value of $\ftwo$ in the caricature}
\+{\F}{F}{an arbirary function}
\+{\GG}{G}{Inverse of $\FF$}
\+{\GNaa}{g}{non-dimensional $\GNa$ }
\+{\GNa}{G_{Na}}{normalized max fast inward (Na) conductance}
\+{\Gtilde}{\tilde\G}{Caricature of $\G(\E)$}
\+{\Gtwol}{g_{21}}{nonzero value of $\gtwo(\E<\Edagger)$ in the caricature}
\+{\Gtwor}{g_{22}}{nonzero value of $\gtwo(\E\ge\Edagger)$ in the caricature}
\+{\G}{G}{time-independent current in modified model}
\+{\HHlwr}{\underline{\HH}}{lower bound for $\GNaa$ }
\+{\HH}{L}{Dispersion function}
\+{\Heav}{\theta}{Heaviside function}
\+{\IKi}{I_{\mathrm{K}_1}}{a potassium current}
\+{\INac}{I_{\mathrm{Na}_c}}{time-independent component of the inward (Na) current}
\+{\INa}{I_{\mathrm{Na}}}{time-dependent component of the fast inward (Na) current}
\+{\Il}{I_{\l}}{individual ionic current number $\l$}
\+{\Intv}{\mathcal{I}}{an interval of parameter values}
\+{\Islow}{I_{\Sigma}}{all slow currents}
\+{\Is}{I_s}{slow inward (Ca) current}
\+{\Ixi}{I_{x_1}}{another potassium current}
\+{\LambertW}{W}{Lambert's $W$-function}
\+{\Nss}{\mathcal{N}}{equation of the superslow manifold}
\+{\N}{N}{number of abstract dyn eqns}
\+{\RR}{r}{ratio of modified Bessel I's with successive indices}
\+{\SS}{S}{$\so$ as function of $\GNaa$}
\+{\T}{T}{fast time, $\t/\eps$}
\+{\Valpasth}{\V_{\alpha}^{2}}{Higher of the crit values of $\Valpha$, super-caricature}
\+{\Valpastl}{\V_{\alpha}^{1}}{Lower of the crit values of $\Valpha$, super-caricature}
\+{\Valpha}{\V_{\alpha}}{non-dimensional \Ealpha}
\+{\Vdagger}{\V_{\dagger}}{non-dimensional \Edagger}
\+{\Vomega}{\V_{\omega}}{Postfront voltage}
\+{\Vone}{\V_1}{dimensionless of $\Eone$}
\+{\V}{v}{non-dimensional voltage}
\+{\X}{X}{Zoomed space}
\+{\Z}{Z}{Zoomed travelling wave coord}
\+{\arb}{\theta}{integration constant in the slow solution}
\+{\auxvar}{\sigma}{a generic auxiliary variable with short scope}
\+{\a}{a}{Barkley a prameter}
\+{\b}{b}{Barkley b prameter}
\+{\cb}{\c_b}{Back speed}
\+{\cfunb}{\cfun_b}{Back speed as a function of the slow variable}
\+{\cfunf}{\cfun_f}{Front speed as a function of the slow variable}
\+{\cfun}{\mathcal{C}}{Wave speed as a function of the slow variable}
\+{\cf}{\c_f}{Front speed}
\+{\cinterval}{\Intv_{\c}}{interval $(\cmin,\cmax)$}
\+{\cmax}{\c_{\mathrm{max}}}{maximal possible $\c$}
\+{\cmin}{\c_{\mathrm{min}}}{minimal possible $\c$}
\+{\c}{c}{Wave speed}
\+{\dgate}{d}{a $\Is$ gate}
\+{\epstwo}{\epsilon_2}{the small parameter of the secondary embedding ($\n$ is slow)}
\+{\eps}{\epsilon}{the small parameter of the main embedding}
\+{\ff}{f}{left-hand side equation of the absolute minimum of $\HHlwr$}
\+{\fgate}{f}{another $\Is$ gate}
\+{\fhnalp}{\alpha}{Slow dynamics speed parameter in FHN model, NPK formulation}
\+{\fhnsq}{\auxvar}{an aux fn of $\fhnth$}
\+{\fhnth}{\beta}{Threshold parameter in FHN model, NPK formulation}
\+{\fu}{F_\u}{rate of change of $\u$ (tihkonov fast) variable}
\+{\fv}{F_\v}{rate of change of $\v$ (tikhonov slow) variable}
\+{\fy}{\mx{F}_{\y}}{rate of change of $\y$ variable(s)}
\+{\gNa}{g_{\mathrm{Na}}}{max $\INa$ conductance}
\+{\ggam}{\gamma}{in stead of $\C^2$}
\+{\hbar}{\Qstat{h}}{fast inward current inactivation gate quasistationar}
\+{\hfast}{H}{h value in the fast subsystem}
\+{\h}{h}{fast inward current inactivation gate}
\+{\jalp}{\j_{\alpha}}{value of $\j$ at the front in the slow pblm}
\+{\jbar}{\Qstat{j}}{fast inward current slow inactivation gate quasistat}
\+{\jfast}{J}{value of $\j$ in the fast pblm}
\+{\jmin}{\j_{\min}}{min value of $\jfast$ needed for front existence}
\+{\jp}{i}{piece index in the slow soln}
\+{\j}{j}{fast inward current slow inactivation gate}
\+{\kone}{k_1}{lower slope of $\Gtilde$}
\+{\kthree}{k_3}{upper slope of $\Gtilde$}
\+{\ktwo}{k_2}{middle slope of $\Gtilde$}
\+{\l}{l}{summation index in various places}
\+{\mbarcub}{m^3_\infty}{same as $\mbar^3$}
\+{\mbar}{\Qstat{m}}{fast inward current activation gate quasistationar}
\+{\m}{m}{fast inward current activation gate}
\+{\ndel}{m}{dev from qstat val of $\n$ in a slow soln piece}
\+{\ndrei}{\overset{3}{\n}}{third piece for slow CN $\n$ soln}
\+{\ninf}{\delta}{qstat val of $\n$ in a slow soln piece}
\+{\njot}{\overset{\jp}{\n}}{$\jp$-piece of CN slow $\n$ soln}
\+{\nmax}{\n_{\max}}{max value of $\n$ over the period}
\+{\nstar}{\n_*}{min value of $\n$ over the period}
\+{\nzwei}{\overset{2}{\n}}{second piece of CN slow $\n$ soln}
\+{\n}{n}{slow outward current activation gate}
\+{\pp}{p}{aux var in Proposition \ref{prop-one}}
\+{\p}{\fu}{N-shaped fast dynamics in general theory}
\+{\q}{\fv}{slow dynamics in general theory}
\+{\so}{\s_0}{the $\za=0$ value of $\s$}
\+{\s}{s}{auxiliary independent variable}
\+{\tast}{\zt_{\ast}}{in the caricature model, the time at which $\E(\tast)=\East$}
\+{\tauE}{\tau_{\E}}{relaxation timescale of $\E$}
\+{\tauh}{\tau_{\h}}{relaxation timescale of $\h$}
\+{\tauj}{\tau_{\j}}{relaxation timescale of $\j$}
\+{\taul}{\tau_{\l}}{relaxation timescale of $\l$-th var}
\+{\tauy}{\tau_{\y}}{relaxation timescale of $\y$}
\+{\ta}{\tau}{non-dimensional time}
\+{\tdagger}{\zt_{\dagger}}{in the caricature model, the time at which $\E(\tdagger)=\Edagger$}
\+{\tildegtwo}{\tilde{g_2}}{maximal negative slow outward (K) current}
\+{\ttwo}{\zt_2}{duration of the second half of the AP trajectory}
\+{\t}{t}{original (slow) time}
\+{\ualpb}{\ufast_{\alpha}^{b}}{pre-back $\u$}
\+{\ualpf}{\ufast_{\alpha}^{f}}{pre-front $\u$}
\+{\ualpha}{\ufast_\alpha}{pre-front value of $\u$ ? }
\+{\ufastb}{\ufast_b}{fast solution for the back}
\+{\ufastf}{\ufast_f}{fast solution for the front}
\+{\ufastv}{\ufast_0}{pinning value of $\ufast$}
\+{\ufast}{\Efast}{excitation variable in Tikhonov system - fast-time}
\+{\ufun}{\mathcal{\Efast}}{Wave profile as a function of the fast coord}
\+{\ulft}{\u_l}{left asymptotic value of $\u$}
\+{\umin}{\u_{-}}{left piece of $\u$-nulcline}
\+{\uomega}{\ufast_\omega}{post-front value of $\u$ ? }
\+{\uomgb}{\ufast_{\omega}^{b}}{post-back $\u$}
\+{\uomgf}{\ufast_{\omega}^{f}}{post-front $\u$}
\+{\uplumin}{\u_{\pm}}{the two stable branches of $\u$-nulcline}
\+{\uplu}{\u_{+}}{right piece of $\u$-nulcline}
\+{\urgt}{\u_r}{right asymptotic value of $\u$}
\+{\uslow}{\bar{\u}}{slow subsystem approximation of the Tikhonov solution}
\+{\uth}{\u_{\ast}}{thershold piece of $\u$-nulcline}
\+{\u}{E}{excitation variable in Tikhonov system}
\+{\vb}{\v_b}{the value of the slow var at the back}
\+{\vfastb}{\vfast_b}{the value of the slow var in the fast system at the back}
\+{\vfastfb}{\vfast_*}{$\vfastf$ or $\vfastb$ as the case may be}
\+{\vfastf}{\vfast_f}{the value of slow var in the fast system at the front}
\+{\vfast}{Y}{recovery variable in Tikhonov system - fast-time }
\+{\vfb}{\v_*}{$\vf$ or $\vb$ as the case may be}
\+{\vf}{\v_f}{the value of slow var at the front}
\+{\vinterval}{\Intv_{\v}}{interval $(\vmin,\vmax)$}
\+{\vmax}{\v_\mathrm{max}}{max of $\v$ range vrt $\p$ properties in Tikhonov systems}
\+{\vmin}{\v_\mathrm{min}}{min of $\v$ range vrt $\p$ properties in Tikhonov systems}
\+{\vslow}{\bar{\v}}{slow subsystem approximation of the Tikhonov solution}
\+{\vsubinterval}{\Intv'_{\v}}{interval $(\vsubmin,\vsubmax)$}
\+{\vsubmax}{\v'_\mathrm{max}}{min of $\v$ range vrt $\q$ properties in Tikhonov systems}
\+{\vsubmin}{\v'_\mathrm{min}}{min of $\v$ range vrt $\q$ properties in Tikhonov systems}
\+{\v}{y}{recovery variable in Tikhonov system}
\+{\w}{w}{abstract dynamic variable; aux var \eq{change-to-bessel}}
\+{\xa}{\xi}{non-dimensional space}
\+{\xone}{x_1}{the x1 gate}
\+{\x}{x}{space variable}
\+{\ybar}{\Qstat{\y}}{quasistat value of $\y$}
\+{\y}{\mx{y}}{all (all the other) gating variables}
\+{\zadagger}{\kappa}{point at which $\V=\Vdagger$}
\+{\za}{\zeta}{CN nondimensional travel coordinate}
\+{\zfb}{\z_*}{front or back position in the travel coord}
\+{\zt}{\eta}{CN time-like travel coordinate}
\+{\zz}{\Z}{inner z?}
\+{\z}{z}{Travelling wave coordinate}

\begin{document}

\title{Asymptotics of conduction velocity restitution  in models of
electrical excitation in the heart{}}
\author{R.~D.~Simitev$^1$ and V.~N.~Biktashev$^2$}
\maketitle

$^1$ Department of Mathematics,
     University of Glasgow, 
     Glasgow G12 8QW, UK

$^2$ Department of Mathematical Sciences,
     University of Liverpool, 
     Liverpool L69 7ZL, UK 
     
\begin{abstract}
  We extend a non-Tikhonov asymptotic embedding,
  proposed earlier, for calculation of conduction velocity restitution
  curves in ionic models of cardiac excitability. Conduction velocity restitution is
  the simplest nontrivial spatially extended problem in excitable
  media, and in the case of cardiac tissue it is an important tool for
  prediction of cardiac arrhythmias and fibrillation.  An idealized
  conduction velocity restitution curve requires solving a nonlinear
  eigenvalue problem with periodic boundary conditions, which in the
  cardiac case is very stiff and calls for the use of asymptotic
  methods. We compare asymptotics of restitution curves in four
  examples, two generic excitable media models, and two ionic cardiac
  models. The generic models include the classical FitzHugh-Nagumo
  model and its variation by Barkley. They are treated with standard
  singular perturbation techniques. The ionic models include a
  simplified ``caricature'' of the Noble (1962) model and the Beeler
  and Reuter (1977) model, which lead to
    non-Tikhonov problems where known asymptotic results do not
    apply. The Caricature Noble model is 
  considered with particular care to demonstrate the well-posedness of
  the corresponding boundary-value problem. The developed method
  for calculation of conduction velocity restitution
  is then applied to the Beeler-Reuter model. We discuss new
  mathematical features appearing in cardiac ionic models and possible
  applications of the developed method.
\end{abstract}

\textbf{Keywords} 
action potential;
traveling wave.

\setcounter{tocdepth}{1}
\tableofcontents
\pagebreak

\section{Introduction}
\label{intro}

\paragraph{Cardiac excitability models}
Hodgkin and Huxley's model of the electric properties of the
giant squid axon ~\cite{Hodgkin-Huxley-1952} was the first to  
describe in mathematical terms the exclusively biological
phenomenon of excitability. It started a revolution in science
well-worth the Nobel prize it  
was awarded. This achievement has been followed by the development
of a long sequence of mathematical models of heart excitability
starting from Noble's works~\cite{Noble-1960,Noble-1962}. Due to its
importance for biomedical applications, particularly for understanding 
and treatment of cardiac arrhythmias caused by pathologies of
electrical excitation and propagation, the mathematical
modelling direction has been under intensive development during
the last decades, and currently has reached clinical
applications and industrial scale. It lies in the heart of the
ambitious Physiome project~\cite{Physiome} which aims at a
mathematical description of the physiology of whole organisms. Due to
complexity of the models involved they are mainly used in numerical
computations and contribute a substantial load on the UK national
supercomputer facilities~\cite{Plank-etal-2006}.  

\paragraph{Stiffness of cardiac excitability models}
The computational complexity of cardiac models lies not only in the
complexity of the heart as a system, which compared to the brain is
relatively modest, but also in the essential stiffness of
cardiac equations. These equations have to describe very
sharp and fast excitation fronts where some processes happen on the
scale of tens of microseconds and 
micrometers, through to tissue and organ level, on the scale of
seconds and centimeters, thus covering several orders of magnitude.
Thus a challenge for applied mathematics is how to turn this 
stiffness from an adversary into an ally. A standard approach is to
treat small parameters, responsible for such stiffness, using
asymptotic rather than numerical methods. For the Hodgkin-Huxley
model, a simple caricature easily treatable mathematically has been
introduced by FitzHugh~\cite{FitzHugh-1961} and Nagumo
\etal~\cite{Nagumo-etal-1962}, which was based on a modification of
the classical van der Pol system of
equations~\cite{vanderPol-1920}. Asymptotic analysis of FitzHugh-Nagumo
type systems, a nice
summary of which can be found \eg\ in~\cite{Tyson-Keener-1988}, has
achieved remarkable success in describing, in qualitative terms, many
of the phenomena observed in more realistic, experiment based ionic
models. 

\paragraph{The traditional asymptotic approach}
The essence of the approach is separation of the dynamic variables
into ``fast and slow'', similar to the classical Tikhonov-Pontryagin
scheme~\cite{Tikhonov-1952,Pontryagin-1957,Mishchenko-Rozov-1980}, only in
a spatially extended context. A typical solution consists of moving,
fast and steep ``fronts'' and ``backs'' of excitation pulses,
located near codimension-one manifolds, \ie\ points in one spatial
dimension (1D), lines in two spatial dimensions (2D), and surfaces in
three spatial dimensions (3D), which are interspersed by smooth and slow
intervals.  During the fast fronts and backs, the slow variables
remain almost unchanged. During the slow intervals, the fast variables
remain very close to their quasi-stationary values determined by the
current values of the slow variables. The slow pieces are typically of
two kinds: with lower and with higher values of the transmembrane
voltage or a variable that corresponds to it. The lower-voltage,
``diastolic'' pieces are close to or include the ``resting state'',
representing excitable tissue which was not excited for a long time,
and the higher-voltage, ``systolic'' pieces represent the ``action
potential'' phase of the excitation. An extra feature of 2D and 3D is
the possibility of ``wave breaks'', which are of particular relevance
for cardiac arrhythmias. Such wave breaks are moving codimension-two
manifolds, \ie\ points in 2D and lines in 3D, where the fronts and
backs meet. It is essential that mathematically, fronts and backs are
objects of the same nature, differing only in the direction of motion:
at the fronts, the systolic phase advances; at the backs, the systolic
phase recedes, so the wave break is where the interface between the
phases is momentarily stalled. This description allows even some
analytical treatment of the motion of wave breaks in 2D, including
steadily rotating and meandering spiral waves of
excitation~\cite{Hakim-Karma-1999}. Conceivably, this asymptotic
description could be also used numerically within an appropriate
moving interface methodology.

\paragraph{The need for a non-Tikhonov embedding}
However, since models of FitzHugh-Nagumo type have been typically
postulated rather than derived from realistic ionic models of
cardiac excitation, the question about their quantitative
validity was not usually posed. Successful attempts to apply the same
singular perturbation technique as developed for systems of
  FitzHugh-Nagumo type, directly to detailed ionic models, have been
made, \eg~\cite{Keener-1991}, but this did not turn into a mainstream 
practical approach. We believe that the reason is that
systems of FitzHugh-Nagumo type are actually quite different, in
the asymptotic sense, from detailed ionic cardiac models, as they fail
completely to describe, even at a qualitative level, some important
properties of cardiac excitation, such as 
\begin{itemize}
\item slow reporarization,
\item slow subthreshold response,
\item fast accommodation,
\item variable peak voltage and
\item front dissipation, 
\end{itemize}\noindent%
all of which are experimentally well-established
and also successfully reproduced by detailed cardiac ionic models 
\cite{Biktashev-2002,Biktashev-etal-2008}. The slow repolarization
means that, although cardiac excitation pulses do indeed possess steep
fronts, they have no steep backs, at least not steep enough
compared  to the steepness of the fronts, anyway. Hence
interpretation of wave breaks in 2D and 3D as loci where fronts
meet backs is inapplicable to cardiac models for absence of
backs. Since propagation blocks and wave breaks are very
important in most applications of mathematical cardiology, there is no
much hope that the FitzHugh-Nagumo ideology could lead to a
practical numerical tool that could tame the stiffness of the cardiac
equations.   

In a recent series of works~\cite{%
  Biktashev-2002,%
  Biktashev-Suckley-2004,%
  Biktasheva-etal-2006,%
  Biktashev-etal-2008%
}, we have developed an analytical approach to cardiac equations based
on their special structure, different from the FitzHugh-Nagumo
paradigm, and taking into account small parameters actually present
(sometimes hidden) in the equations, rather than trying to force them
into the Procrustean bed of a classical scheme. 
Using the existence of large (or small)
values of 
some variables for model reduction and perturbation analysis is a
basic technique in applied mathematics. One well-known example is
the Quasi(Pseudo)-Steady-State approximation \cite{Heineken-etal-1967,Segel-Slemrod-1989},
which reduces the equations of an enzyme reaction to a singular
perturbation Tikhonov problem. Another prominent example is the wide
application of scale 
separation and model reduction techniques to problems of chemical
kinetics and reactive flows. Since in such problems the number of
reacting species is huge, this is typically done by computational
algorithms such as Computational Singular Perturbations (CSP),
Intrinsic Low-Dimensional Manifolds (ILDM), the Grad Moment method and
others \cite{Gorban-Karlin-2005,Kaper-2002}. It has been shown that these computational
techniques generate the asymptotic expansion of a slow invariant
manifold of a Tikhonov problem \cite{Kaper-2002,Kaper-2004}.

However, the application of asymptotic embedding techniques is not
restricted to Tikhonov problems, nor must it, \apriori, lead to
such. Indeed, our recent works  
\cite{%
  Biktashev-2002,%
  Biktashev-Suckley-2004,%
  Biktasheva-etal-2006,%
  Biktashev-etal-2008%
}
have clearly demonstrated that, to achieve a
physiologically correct asymptotics in realistic models of cardiac
excitation, a parameter embedding is needed which involves a large factor in front
of individual terms, but not the whole, of the right-hand side of some
equations (\eg\ the $\INa$ term in the transmembrane voltage evolution
equation), non-analytical, perhaps even discontinuous, asymptotic
limit of some right-hand sides (\eg\ the $\INa$ gating variables),
even though the original system is analytical, non-isolated equilibria
in the fast subsystem and dynamic variables which change their
character from fast to slow within one solution (\eg\ the
transmembrane voltage).

In particular, we have demonstrated that separate
consideration of the fast subsystem describing excitation front
produces a simple useful criterion of dynamic propagation block in a
modern cardiac model~\cite{Simitev-Biktashev-2006}, and the fast and
slow subsystems can be successfully matched to describe the action
potential as a singular limit in a single-cell (0D) variant of a
simple cardiac model~\cite{Biktashev-etal-2008}. The next step is to
combine the fast and slow description in a spatially extended
context. 

\paragraph{CV restitution curves: a spatiotemporal problem
 involving   fast and slow scales}
The aim of the current work is to make this next step. For this
purpose, we have chosen the simplest nontrivial spatially extended
problem that depends both on the fast and the slow processes: the
conduction velocity restitution curve.  This choice is also motivated
by the practical importance of restitution curves, which are of two
kinds. The action potential duration (APD) restitution curve is the
dependence of the APD on the duration of the preceding diastolic
interval (DI).  Nolasco and Dahlen~\cite{Nolasco-Dahlen-1968} noted
that in a single-cell setting and with a fixed period of
excitation, a slope of the APD(DI) curve greater than one
indicates instability of the even APD sequence. For
this reason the restitution curves are considered an important
instrument in understanding instabilities of excitation waves leading
to onset of cardiac arrhythmias. Later studies have demonstrated that
in a spatially extended context, another important tool is the
conduction velocity (CV) restitution curve, which describes the
dependence of CV on the preceding diastolic interval. The
CV(DI) dependence together with the APD(DI) dependence and the fact
that the overall period known in electrophysiology as Basic Cycle
Length is given by BCL=APD+DI, makes it possible to define the
CV(BCL) dependence, \ie\ relationship between the period of excitation
waves and their propagation speed, which is also known as the
dispersion relation in general wave theory. The CV(DI) curve
depends on the definition of the boundary between action potential phase and the
diastolic phase, which for cardiac excitation pulses is arbitrary for
lack of sharp backs. The CV(BCL) dependence is, on the contrary,
free from such arbitrariness and is well defined mathematically. So in
our study, we shall use this dependence as the restitution curve.

\paragraph{Types of restitution curves}
In view of the clinical importance of fibrillation, numerous
experimental, \eg~\cite{%
  Garfinkel-etal-2000,%
  Watanabe-etal-1995,%
  Boyett-Jewell-1978,%
  Nolasco-Dahlen-1968,%
  Chialvo-etal-1990%
} and numerical \eg~\cite{%
  Courtemanche-1996,%
  Karma-1994,%
  Courtemanche-Glass-Keener-1993,%
  Ito-Glass-1992,%
  Karma-etal-1994,%
  Watanabe-etal-2001%
} studies 
are concerned with tests of this hypothesis, and with
  measurements and computation of restitution curves
in various types of cardiac cells.  Measurement and computation of
restitution curves are not straightforward.  A number of different
experimental/numerical protocols are in use (see \eg~\cite{Schaeffer-etal-2007} and
references therein), which produce different curves and it is not
always
clear which is the most relevant one in a particular case.
For instance, in the so called
``dynamic'' protocol the tissue is paced at a given basic cycle length
until a periodic regime is established, and the APD, DI and CV of the
established pulses are recorded. Then the process is repeated with
other cycle lengths.  Another protocol is the ``S1-S2'' restitution
protocol, in which the tissue is paced at a fixed cycle length S1
until a periodic regime is reached, and is then perturbed by an
out-of-sequence stimulus (S2) and the response is recorded. The
preparation is then paced at a the same S1 until steady-state has been
reached again, and is then perturbed by a different S2. The curve so
measured depends on the choice of the S1 cycle length, and therefore
it is not even unique. Although used in electrophysiological practice,
these protocols have a number of drawbacks: they contain some
arbitrariness and thus lead to results which are not unique, they are
prone to systematic errors since it not easy to distinguish the
ultimate periodic regime from transient, they are time consuming and,
in the case of numerical simulations, computationally expensive since
a repeated solution of large systems of stiff nonlinear partial
differential equations is required. 

\paragraph{The aim of this study}
In the present study we have chosen to use an idealized
definition of the ``dynamic'' restitution protocol, \ie\ we consider
strictly periodic wave solutions, and study the dependence between the
BCL and CV of such solutions. This dependence is well defined
mathematically via 
solvability of the corresponding boundary-value problem with periodic
boundary conditions. This idea is not new but so far it has had only
limited application for the following two reasons. First, the
resulting boundary value problem is typically very stiff, with very
steep upstroke but slow prolonged plateau and recovery stages of a
typical cardiac action potential, and so its direct solution requires
considerable effort.  Secondly, the Tikhonov
asymptotic embeddings which are typically used to alleviate such
scale disparities fail to produce results which are even qualitatively
correct, as noted above. 

We wish to emphasize that the periodic boundary value problem approach
we advocate here is applicable both to cardiac models with
Tikhonov as well as with non-Tikhonov asymptotic
structure and in this work we illustrate both of these cases. However,
in the absence of a rigorous theory of the
non-Tikhonov case, we  make a special effort 
to investigate whether the resulting asymptotic boundary value
problem is well-posed. This is not obvious \apriori.

\paragraph{Structure of the paper}
In section \ref{BVP} we
formulate the periodic boundary-value problem which gives a general
method of computing CV restitution curves regardless of the asymptotic
structure of the particular cardiac model.  In
sections \ref{Tikhonov}, \ref{Barkley} and
\ref{secFHN} we apply the method to well-known models with
Tikhonov asymptotic structure in order to provide
simple illustrations. 
Section \ref{asymptotic} is central to the article. Here
we use a suitably reformulated version of the Noble model of cardiac
Purkinje fibers~\cite{Noble-1962} to illustrate the
non-Tikhonov asymptotic reduction in detail and to
investigate whether the resulting asymptotic boundary value problem is, indeed,
well-posed. In section \ref{applicationBR}, we calculate the full and
the asymptotic CV restitution curves of the Beeler-Reuter ventricular
model~\cite{Beeler-Reuter-1977} and demonstrate a good quantitative agreement.
Section \ref{discussion} provides concluding remarks and suggests
possible extensions of the work.

\section{Restitution curves: the boundary-value problem formulation}
\label{BVP}

A typical voltage-gated model of cardiac excitation and propagation
in a one-dimensional, homogeneous and isotropic medium has the form
of a reaction--diffusion system,
\begin{subequations}
\eqlabel{general}
\begin{gather}
\df{\E}{\t} = \sum_{\l} \Il(\E, \y) + \D\,\ddf{\E}{\x}, \\
\df{\y}{\t}= \fy(\E,\y),
\end{gather}
\end{subequations}
where $\x$ is the spatial coordinate, $\t$ is the time, 
$\E$ is transmembrane
voltage of the cardiocytes, the functions $\Il(\cdot)$ represent
individual transmembrane ionic currents, each conducted by a specific
type of transmembrane channel, the vector $\y$ includes a number of
``gating'' variables controlling the permittivity of the ionic channels
and the intra- and extracellular concentrations of ions involved, and 
$\D>0$ is a ``voltage diffusion constant'', depending on the electric
capacitance of cardiocytes and Ohmic contacts between them.
Note that $\D$ can be made equal to any positive value by rescaling
the spatial variable $\x$; we shall choose this scaling so that
$\D=1$, or related to the small parameter when considering
asymptotics. 
  This means that the dimensionality of $\x$ is that of
  $\t^{1/2}$. Correspondingly,   to compare our subsequent
  results with experimental data,
  lengths and speeds should be scaled up by the factor of $\D^{1/2}$, where
  $\D$ is the value of the voltage diffusion coefficient, dependinging
  on the properties of the given tissue and the direction of wave
  propagation. 

The number of gating variables, concentrations and the form of the
functions $\Il(\cdot)$ and $\fy(\cdot)$ are fitted to reproduce the very latest
experimental observations. As experimental methods improve, the models
evolve to be ever more complicated but the general form of the
reaction--diffusion system \eq{general} has hardly changed since 1962
when the first cardiac model was published by Noble~\cite{Noble-1962}.
A relatively recent but by no means ultimate list of cardiac models
can be found in the review~\cite{Clayton-2001} and confirms this
assertion.

CV restitution curves are typically computed by direct numerical
simulation of the partial differential equations \eq{general}
following a particular protocol. 
As argued above this is computationally expensive and
time consuming, and prone to systematic errors.
A more sound mathematical approach is to look for
solutions in the form of waves travelling with a constant velocity $\c>0$
and a fixed shape. This is guaranteed by the travelling wave ansatz $\F(\z)
= \F(\x-\c\t)$ for the dynamical variables $\F = \E, \y$, where
$\z=\x-\c\t$. Equations \eq{general} are then reduced to a system of
autonomous ordinary differential equations and the CV restitution
curve can be found from the periodic boundary value problem 
\begin{subequations}
\eqlabel{full}
\begin{gather}
  \DDf{\E}{\z} + \c\Df{\E}{\z} +\sum_{\l} \Il(\E, \y)=0, \\
  \c\Df{\y}{\z} + \fy(\E,\y) = 0, \\
  \E(0) =\E(\c\,\BCL), \quad 
  \left.\Df{\E}{\z}\right|_{\z=0} =
  \left.\Df{\E}{\z}\right|_{\z=\c\,\BCL}, \nonumber \\
  \y(0)=\y(\c\,\BCL), \quad
  \E(0)=\Ev, 
\end{gather}
\end{subequations}
where $\BCL$ is the temporal period of the waves. The last boundary condition
$\E(0)=\Ev$ is related to the translational invariance of the
problem. This condition allows the selection of a single solution out of
a one-parametric family of solutions differing from each other
only in their position along the $\z$ axis; thus the exact choice of
$\Ev$ is not essential,
as long as it is selected  within the range of values of $\E(\z)$.
Problem \eq{full} is of order $(\dim(\y)+2)$, where
$\dim(\y)$ is the dimension of the vector $\y$, and its general
solution includes $(\dim(\y)+2)$
arbitrary constants. In addition, the
problem involves two unknown parameters, $\c$ and $\BCL$.  On the
other hand, it has $(\dim(\y)+2)$ periodic boundary conditions plus
the last ``phase'' or ``pinning'' condition required to eliminate the translational
invariance of the system. Thus,  we have $(\dim(\y)+4)$ parameters and
$(\dim(\y)+3)$ constraints on them,  so the solution of the problem
should yield, in principle, a one-parameter family of solutions.
A projection of this family onto the $(\BCL,\c)$ plane is the sought
after ``ideal'' dynamic CV restitution curve describing the
dependence of the wave speed on the wave period.

\section{Outline of the singular perturbation theory of Tikhonov excitable
  systems}
\label{Tikhonov}

The method outlined in section \ref{BVP} is applicable to
any cardiac model but due to the inherent stiffness of cardiac
equations solution it is difficult for a numerical study 
if asymptotics are not exploited. Our first illustrations of the
method will involve the Barkley model \cite{Barkley-1991} 
and the FitzHugh-Nagumo system \cite{FitzHugh-1961,Nagumo-etal-1962}. 
We take advantage of the fact that these models have  well-known
asymptotic structures of a Tikhonov type which has been studied in a
number of works, \eg~%
\cite{%
  Keener-1980,%
  Tyson-Keener-1988,%
  Dockery-Keener-1989,%
  Murray-1991,%
  Keener-Sneyd-1991,%
  Meron-1992%
}
and thus they provide simple illustrations and set the context for our main
results. 

\subsection{Asymptotic reduction of the CV restitution boundary value
  problem for Tikhonov systems}

A typical model with a Tikhonov asymptotic structure has the form
\begin{subequations}                                        
\eqlabel{fhntype}
\begin{gather}
\df{\u}{\t} = \frac{1}{\eps}\p(\u,\v) + \eps \ddf{\u}{\x},      
       \eqlabel{fhntype1}\\
\df{\v}{\t} = \q(\u,\v),  
        \eqlabel{fhntype2}
\end{gather}
\end{subequations}
where $\u$ is interpreted as the voltage while $\v$ is taken to
represent all other variables of an ionic cardiac model.  We
assume, for simplicity, that the dynamical variables $\u$ and $\v$ are
scalar fields, which is true for the Barkley and the FitzHugh-Nagumo
models.  The small parameter $\eps\ll1$ specifies the asymptotic
structure of the system explicitly by indicating the relative
magnitude of the various terms in the model.
Here and in subsequent asymptotic formulations, the spatial scaling
is chosen so that the diffusion coefficient is equal to $\eps$; the
convenience of this choice will be evident shortly.
In order for equations
\eq{fhntype} to have excitable or oscillatory dynamics, certain
properties of the functions $\p(\cdot)$ and $\q(\cdot)$ need to be
assumed. We shall assume that in a certain interval of $\v$ values,
$\v\in\vinterval=(\vmin,\vmax)$, the following is true. 
\begin{assm}
\item\label{A-three-roots} Equation $\p(\u,\v)=0$, 
understood as an equation for $\u$ at a fixed $\v$, has three roots, namely
  $\umin(\v)<\uth(\v)<\uplu(\v)$. 
This equation defines the reduced slow manifold of
  the system \eq{fhntype} in the sense of the geometric singular
  perturbation theory \cite{JonesCKRT,Kaper1999}.
\item\label{A-stability} The roots have alternating stability in linear approximation,
  that is,
  $\@_\u\p\left(\umin(\v),\v\right)<0$, 
  $\@_\u\p\left(\uth(\v),\v\right)>0$ and
  $\@_\u\p\left(\uplu(\v),\v\right)<0$, 
where $\@_\u$ denotes a partial derivative.
\end{assm}
  Further, we assume that, in a possibly smaller interval
  $\v\in\vsubinterval=(\vsubmin,\vsubmax)\subseteq\vinterval$, the following is true.
\begin{assm}
\item\label{A-direction} The slow dynamics for $\v$ is growth
  for the lower root $\u=\umin(\v)$ and decrease for the upper root $\u=\uplu(\v)$, that is
  $\q(\umin(\v),\v)>0$ and $\q(\uplu(\v),\v)<0$.
\item\label{A-closed} The periodic wave solutions of interest only
  involve the interval $\v\in\vsubinterval$.
\end{assm}
These assumptions are true for the Barkley and the
FitzHugh-Nagumo models, and are illustrated in \figs{bkl}(a) and
\figref{fhn}(a) below. The last assumption~\ref{A-closed}, unlike the
first three, is difficult to formulate in \apriori\ terms, and we shall
discuss its implications as we obtain the relevant results below.

Due to assumption~\ref{A-three-roots}, the sets
$\left(\umin(\v),\v)\right)$ and $\left(\uplu(\v),\v)\right)$ are
disjoint in the $(\u,\v)$-plane. These two sets are known as the
\emph{diastolic} and the \emph{systolic} branches of the reduced slow
manifold, respectively.

To formulate the CV restitution boundary value problem
\eq{full} for equations \eq{fhntype}, we look for solutions in the
form of waves travelling with a constant velocity $\c$ and a fixed
shape \ie\ we assume the travelling wave ansatz  %
  $\z=\x-\c\t$, which gives the asymptotic boundary-value problem
  \begin{subequations}
    \begin{gather}
      \eps^2\DDf{\u}{\z} + \eps\c\Df{\u}{\z} + \p(\u,\v) = 0,           \eqlabel{fhntype-rc1}\\
      \c\Df{\v}{\z} + \q(\u,\v) = 0,                                    \eqlabel{fhntype-rc2} \\
      \u(0) =\u(\BCL)=\Ev, \quad
      \left.\Df{\u}{\z}\right|_{\z=0} =\left.\Df{\u}{\z}\right|_{\z=\BCL}, \quad
      \v(0)=\v(\BCL) .                                                   \eqlabel{fhntype-rc3}
    \end{gather}
    \eqlabel{fhntype-rc}
\end{subequations}
We first formulate the slow and fast subsystems corresponding to this
problem, and after that we will discuss matching and boundary conditions.

The slow-time subsystem is obtained immediately from equations
\eq{fhntype-rc} in the limit $\eps\to 0$, and has the form
\begin{subequations}
  \eqlabel{slow.fhntype}
\begin{gather}
\p(\uslow,\vslow) =0,
\eqlabel{slow.fhntype1}\\
\c\Df{\vslow}{\z} + \q(\uslow,\vslow)=0 ,
\eqlabel{slow.fhntype2}
\end{gather}
\end{subequations}
so a unique solution is obtained by imposing a
a single boundary condition, \eg
\begin{gather}
\eqlabel{slow.bc}
\vslow(0)=\vfb,
\end{gather}
where $\vfb$ is a constant. 
Here we use $\uslow(\z)=\lim_{\eps\to0}\u(\z)$ and
$\vslow(\z)=\lim_{\eps\to0}\v(\z)$ to denote the slow-subsystem
solution approximation and distinguish it from the exact solution.

The fast-time subsystem is obtained from equations \eq{fhntype-rc} by first
rescaling the traveling wave coordinate, %
$\zz=(\z-\zfb)/\eps$, where $\zfb=\const$ is the position of the jump
(front or back) in the slow wave coordinate,
and then
taking the limit $\eps\to0$, which gives the equations
\begin{subequations}                       
  \eqlabel{fhntype.fast}
\begin{gather}
\DDf{\ufast}{\zz} + \c\Df{\ufast}{\zz}+ \p(\ufast,\vfast)=0,     \eqlabel{fhntype.fast1}\\
\Df{\vfast}{\zz} = 0 ,                               \eqlabel{fhntype.fast2}
\end{gather}
\end{subequations}
the boundary conditions for which can be taken in the form
\begin{gather}
  \eqlabel{fhntype.BC}
  \ufast(-\infty)=\ulft, \quad
  \ufast(+\infty)=\urgt, \quad 
  \ufast(0)=\ufastv, \\ 
  \left.\Df{\ufast}{\zz}\right|_{\zz\to+\infty} = 0, \quad
  \vfast(-\infty)=\vfastfb. \nonumber
\end{gather}

Above we have introduced $\Efast(\zz)=\lim_{\eps\to0}\E(\zfb+\eps\Z)$
and $\vfast(\zz)=\lim_{\eps\to0}\v(\zfb+\eps\Z)$ for the
fast-subsystem approximation to explicitly distinguish it from the
solution in the original slow coordinate.  The arbitrary constant
$\ufastv$ is assumed in the range of $\ufast$, and is used to define
the position $\zfb$ of the jump in terms of the slow wave coordinate
$\z$, so that $\E(\zfb)=\ufastv$.  Note that equations
\eq{fhntype.fast} are obtained in that form without the need of
  $\eps$-dependent scaling of the speed $\c$ only if the spatial
scaling depending on $\eps$ is chosen as in equation \eq{fhntype},
which is the reason for that choice.

\subsection{Solution of the fast subsystem}

Due to equation \eq{fhntype.fast2}, $\vfast$ is a first integral, and
then equation~\eq{fhntype.fast1} together with the boundary conditions
\eq{fhntype.BC} present an eigenvalue problem for the profile
$\ufast(\Z)$ and velocity $\c$ of a trigger wave, depending on
$\vfast=\vfastfb$ as a parameter. It also depends, of course, on
the values of the voltage to the left and to the right of the
front, 
$\ulft$ and $\urgt$, which should be the two stable roots of
$\p(\cdot,\vfastfb)$, \ie\
$\{\ulft,\urgt\}=\{\umin(\vfastfb),\uplu(\vfastfb)\}$. 
Under the assumptions made about function $\p(\cdot,\vfastfb)$ and
with an appropriate choice of the pinning value $\ufastv$, the 
fast-time boundary-value problem for the trigger wave has a unique
solution, which is guaranteed by a result due to Aronson and
Weinberger~\cite[Theorem 4.1]{Aronson-Weinberger-1978}.
We denote this unique solution by
\begin{gather}
  \ufast(\Z) = \ufun(\Z;\vfastfb,\ulft,\urgt);
                                                \eqlabel{fhntype-profile-front}
\end{gather}
and the corresponding propagation speed by
\begin{equation}
  \c = \cfun(\vfastfb;\ulft,\urgt). 
                                                \eqlabel{fhntype-speed-front}
\end{equation}
The boundary value problem \eq{fhntype.fast2}, \eq{fhntype.BC} is
invariant with respect to simultaneous transformation $\Z\to-\Z$,
$\c\to-\c$, $\ulft\leftrightarrow\urgt$. Hence, it follows that
\begin{equation} 
  \ufun(\Z;\vfastfb,\urgt,\ulft) = \ufun(-\Z;\vfastfb,\ulft,\urgt) 
                                                \eqlabel{fhntype-profile-front-back}
\end{equation} 
and
\begin{equation} 
  \cfun(\vfastfb;\urgt,\ulft) = -\cfun(\vfastfb;\ulft,\urgt).
                                                \eqlabel{fhntype-speed-front-back}
\end{equation} 
Note that these formal solutions can be with positive as well as
negative $\c$; we are, however, only interested in the waves
propagating rightwards, $\c>0$. 
Now we can discuss fronts and backs as two different
types of trigger waves. 
\begin{itemize}
\item Suppose that for some $\vfastf$ we have
  $\cfun(\vfastf;\uplu(\vfastf),\umin(\vfastf))>0$.  This means that we
  have a forward propagating trigger wave that switches the system from
  the lower quasi-equilibrium $\urgt=\umin(\vfastf)$ to the upper
  quasi-equilibrium $\ulft=\uplu(\vfastf)$. We will call this type
    of fast solution a front.
\item Now suppose that for some $\vfastb$ we have
  $\cfun(\vfastb;\uplu(\vfastb),\umin(\vfastb))<0$.  This means that an
  up-jump trigger wave does not propagate forwards but retracts
  backwards, and is not suitable for us as we are interested in forward
  propagating waves, $\c>0$. However, due to \eq{fhntype-speed-front},
  we know that we then have
  $\cfun(\vfastb;\umin(\vfastb),\uplu(\vfastb))>0$, that is there is a
  forward propagating down-jump trigger wave switching from the upper
  quasi-equilibrium $\urgt=\uplu(\vfastb)$ to the lower
  quasi-equilibrium $\ulft=\umin(\vfastb)$. We will call this type
  of fast solution a back.
\end{itemize}

\subsection{Solution of the slow subsystem}

Equation \eq{slow.fhntype1} implies that $\uslow=\uplumin(\vslow)$ or
$\uslow=\uth(\vslow)$. By assumption~\ref{A-stability}, the latter
solution is the unstable branch of the reduced slow manifold
$\p(\uslow,\vslow)=0$ while $\uplumin(\vslow)$ are the stable ones.
Hence ignoring the possibility of ``canard'' solutions that involve
the unstable branch, we must solve
\begin{gather}
\eqlabel{slowequ}
\c \Df{\vslow}{\z} +  \q\left(\uplumin(\vslow),\v\right) = 0,
\end{gather}
which is separable and can be easily integrated, giving the (spatial) length of
the piece of a solution say between $\vslow(\z_1)=\v_1$ and
$\vslow(z_2)=\v_2$ as
\begin{gather}
  \z_2-\z_1 = \c \int\limits_{\v_2}^{\v_1} \frac{\d\v}{\q(\uplumin(\v))} ,
                                                            \eqlabel{fhntype-piecelength}
\end{gather}
where the plus subscript refers to the systolic (action potential) branch
and the minus subscript refers to the diastolic (diastolic interval)
branch. Naturally, $\z_2>\z_1$ requires that $\v_1-\v_2$ and
$\q(\uplumin(\v)$ have the same sign. 

\subsection{Matching}

A period of a steadily propagating periodic pulse train, to
which the asymptotics described above are applicable, must include at
least one fast front, one fast back, one systolic interval and one
diastolic interval. The restitution curve sought for can be obtained
from conditions of matching of these four pieces.

The asymptotic matching of the fast and slow pieces in the leading
order in $\eps$ is rather straightforward. Let us consider a fast jump
solution $(\ufast(\Z),\vfast(\Z))$,
existence of which is guaranteed by
the Aronson-Weinberger theorem,
\cite[Theorem 4.1]{Aronson-Weinberger-1978},
located at $\z=\zfb$ so that $\Z=(\z-\zfb)/\eps$,
and compare it with the slow solutions $(\u(\z),\v(\z))$ ahead and
behind it. By van Dyke's matching rule, we have 
  \begin{gather}
    \lim\limits_{\Z\to-\infty}\ufast(\Z) \equiv \ulft 
    = 
    \lim\limits_{\z\to\zfb-0}\uslow(\z), \nonumber \\
    \lim\limits_{\Z\to+\infty}\ufast(\Z) \equiv \urgt
    = 
    \lim\limits_{\z\to\zfb+0}\uslow(\z), \nonumber \\
    \lim\limits_{\z\to\zfb-0}\vslow(\z) = \vfastfb = \lim\limits_{\z\to\zfb+0}\vslow(\z), 
  \end{gather}
that is, $\vslow$ is continuous  across $\z=\zfb$ and the jump of
$\uslow$ at $\zfb$ is related to $\vslow(\zfb)$ and $\c$ via the speed
equation \eq{fhntype-speed-front}.

Note that above we distinguished between $\v$, $\vfast$ and $\vslow$
only in order to demonstrate explicitly the decoupling of the slow-
and the fast-time problems. However, they all coincide in the
leading order in $\eps$, below we will, for simplicity of notation,
use $\v$ to represent any of them, as this distinction does not run
any deeper. This applies, in particular, to $\vfb=\vfastfb$,
$\vf=\vfastf$ and $\vb=\vfastb$. For the same simplicity, henceforth
we write $\u$ instead of $\uslow$ as they coincide in the same limit
almost everywhere. We keep distinguishing $\ufast$, though, as it
describes gradual change of the voltage where $\uslow$ has a jump.

Let us take the conduction velocity $\c$ as the parameter, \ie\
construct the periodic pulse propagating with a given speed $\c>0$,
and then calculate its temporal period $\BCL$. 
Further to assumptions~\ref{A-three-roots}--\ref{A-closed}, we make
the additional assumption.
\begin{assm}
\item\label{A-monot} Function $\cfun(\v;\uplu(\v),\umin(\v))$ is a
  monotonically decreasing function of $\v$.
\end{assm}
\noindent
This is easily verified for the two examples that follow; note that
a monotonically increasing function can be deal with in just the same
way. Then equation
\begin{gather}
  \cfunf(\vf)\equiv \cfun\left( \vf;\uplu(\vf),\umin(\vf) \right)=\c
                                        \eqlabel{fhntype.frontspeed}
\end{gather}
may have at most one solution for first integral parameter $\vf$
of the front, and equation
\begin{gather}
  \cfunb(\vf)\equiv\cfun\left( \vb;\umin(\vb),\uplu(\vb) \right)=\c
                                        \eqlabel{fhntype.backspeed}
\end{gather}
may have at most one solution for the first integral parameter
$\vb$ of the back, and we always have
\begin{gather}
  \vb>\vf,                      \eqlabel{allsmaller}
\end{gather}
as positive values of a monotonically decreasing function are achieved at
smaller values of the arguments than negative values.

Now we can formalize assumption~\ref{A-closed} in the following way.
\begin{enumerate}
  \item[A4.]\label{A-closed1} There is a nonempty interval of $\c$ values,
    $\cinterval=(\cmin,\cmax)$, which is the interval of
    interest, such that
    $\cfunb^{-1}(\cinterval)\subset\vsubinterval$ and
    $\cfunf^{-1}(\cinterval)\subset\vsubinterval$. 
\end{enumerate}
\noindent
  So, under the assumptions \ref{A-three-roots}--\ref{A-monot}, for every
  $\c\in\cinterval$ there are exactly two types of fast jump solutions, 
  \\ $\bullet$ %
  a front, at $\v=\vf$, with a pre-front voltage $\ualpf$ and post-front voltage
  $\uomgf$,  
    \begin{gather}
      \ufastf(\Z) = \ufun(\Z;\vf,\uomgf,\ualpf), \nonumber\\
      \ualpf=\umin(\vf), \quad
      \uomgf=\uplu(\vf), \quad
      \vf=\cfunf^{-1}(\c) ,
    \end{gather}
  \\ $\bullet$ %
  and a back,  at $\v=\vb$, with a pre-back voltage $\ualpb$ and post-back voltage
  $\uomgb$,  
    \begin{gather}
      \ufastb(\Z) = \ufun(\Z;\vb,\uomgb,\ualpb), \nonumber\\
      \ualpb=\uplu(\vb), \quad
      \uomgb=\umin(\vb), \quad
      \vb=\cfunf^{-1}(\c). 
    \end{gather}

Let us now consider the slow pieces. Since there is only a unique choice
of the $\v$-values they can have at their ends, namely $\vb$ and
$\vf$ as shown above, there are only two possibilities: a slow
piece that has $\vb$ on the left (back) end and $\vf$ on the
right (head) end, and \viceversa.
For a piece with $\vb$ on the left and $\vf$ on on the right,
due to inequality~\eq{allsmaller} and assumption~\ref{A-direction},
equation~\eq{fhntype-piecelength} gives positive length of the piece
only if it is a systolic piece. The temporal duration of such piece
is the APD, and equals 
\begin{gather}
  \APD=\int\limits_{\vf}^{\vb} \frac{\d\v}{-\q(\uplu(\v))} .
                                        \eqlabel{fhntype-APD}
\end{gather}
Similarly, for a piece with $\vf$ on the left and $\vb$ on on the right,
due to inequality~\eq{allsmaller} and assumption~\ref{A-direction},
equation~\eq{fhntype-piecelength} gives positive length of the piece
only if it is a diastolic piece. The temporal duration of such piece
is the DI, and equals 
\begin{gather}
  \DI=\int\limits_{\vf}^{\vb}\frac{\d\v}{\q(\umin(\v))} .
                                        \eqlabel{fhntype-DI}
\end{gather}
Hence, we have demonstrated that in the assumptions
made~\ref{A-three-roots}--\ref{A-monot}, for every $\c\in\cinterval$
there is exactly one, up to translations along the $\z$ axis, solution
of each of the following four kinds: a front, a systolic slow piece, a
back and a diastolic slow piece, and they can be matched only in a
unique order. Hence for every $\c$ we have a periodic solution, each
period of which consists of exactly one piece of each kind.

To summarize, for every $\c\in\cinterval$, in the leading order in
$\eps$, the temporal period of the solution is
\begin{gather}
  \BCL=\APD+\DI=\int\limits_{\vf}^{\vb}\left(
    \frac{1}{\q(\umin(\v))} - \frac{1}{\q(\uplu(\v))}
  \right)\,\d\v,
                                                  \eqlabel{fhntype-period}
\end{gather}
where $\vf=\cfunf^{-1}(\c)$ and $\vb=\cfunb^{-1}(\c)$ are the unique
solutions of equations \eq{fhntype.frontspeed} and
\eq{fhntype.backspeed} respectively. 

\subsection{Special case: cubic fast dynamics}

Problem \eq{fhntype.fast1}, \eq{fhntype.BC} has an explicit solution
in two popular special cases, for a
cubic~\cite{Zeldovich-FrankKamenetsky-1938} and for a piece-wise
linear~\cite{McKean-1970} dependencies $\p(\cdot;\v)$.  The two simple
examples that follow fall in the case of cubic nonlinearity,
\begin{gather}
  \p(\u,\v) = 
  - \A \, \left(\u-\umin\right) \, 
    \left(\u-\uth\right) \, \left(\u-\uplu\right), \nonumber\\ 
  \u{\pm,*}=\u{\pm,*}(\v), \quad \A=\A(\v)>0,       \eqlabel{cubic}
\end{gather}
which evidently satisfies assumption~\ref{A-stability} as long as
assumption~\ref{A-three-roots} holds. Assumption~\ref{A-monot} imposes
obvious constraints on the functions $\uplumin(\v)$ defining this
nonlinearity. 

In this case it is convenient to choose
$\ufastv=(\umin+\uplu)/2=(\ulft+\urgt)/2$, and then the solution to
problem \eq{fhntype.fast}, \eq{fhntype.BC} is
\begin{gather}
  \ufast=\ufastv + \left(\urgt-\ulft\right) \, \tanh\left(
      (\A/8)^{1/2} (\urgt-\ulft) \zz
  \right) / 2,                                                
\nonumber \\
  \c=\cfun(\vf;\ulft,\urgt)= (2\A)^{1/2} \left(\ufastv - \uth\right) ,
                                                  \eqlabel{zfkfront}
\end{gather}
and we remind that $\{\ulft,\urgt\}=\{\umin(\vf),\uplu(\vf)\}$. 

\section{Asymptotic restitution curves in the Barkley model}
\label{Barkley}

\dblfigure{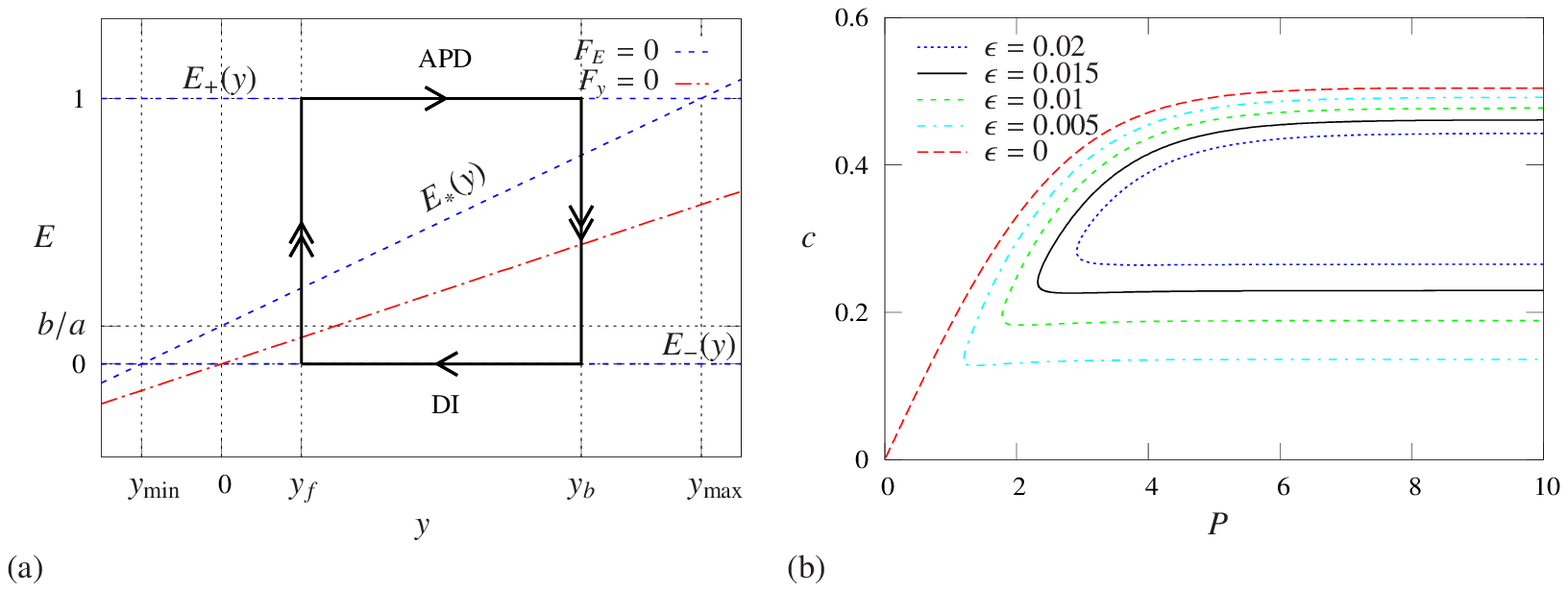}{ 
  (color online)
  CV asymptotics in the Barkley model. Parameters: $\a=0.7$, $\b=0.1$. 
  (a) Schematic of the asymptotic pulse solution in the $(\v,\u)$ plane. 
  Note that in this and the next figures, the direction of the axes
  ($\u$ vertical, $\v$ horizontal) is different from the traditional
  ($\v$ vertical, $\u$ horizontal), which is to comply with subsequent
  ionic gate figures where the transmembrane voltage (corresponding to
  $\u$ here) is on the vertical axis.
  (b) CV restitution curves for various values of $\eps$. The curve $\eps=0$
  is the asymptotic given by \eq{barkley-rc}. 
}{bkl}

\subsection{The model}

The functions $\p(\cdot,\cdot)$ and $\q(\cdot,\cdot)$ of the Barkley
model \cite{Barkley-1991} are given by
\begin{subequations}
                                                 \eqlabel{barkley}
\begin{gather}
\p(\u,\v) = \u\,(1-\u)\left(\u-\frac{\v+\b}{\a}\right),
                                                 \eqlabel{barkley1} \\
\q(\u,\v)=\u - \v,
                                                 \eqlabel{barkley2}
\end{gather}
\end{subequations}
where $\a$ and $\b$ are parameters, satisfying
$\b>0$, $2\b<\a<1+\b$.
The equation of the reduced slow manifold
$\p(\u,\v)=0$ is trivial to resolve and yields the branches
\begin{gather}
\umin=0, \quad \uth=(\v+\b)/\a, \quad \uplu=1.
                                                 \eqlabel{barkley.pieces}
\end{gather}
With this choice of the branches, assumptions~\ref{A-three-roots}
and~\ref{A-stability} are 
satisfied for $\v\in\vinterval=(-\b,\a-\b)$. 
However, assumption~\ref{A-direction}, specifically the condition 
$\q(\uplu(\v),\v)<0$, narrows this down to $\v\in\vsubinterval=(0,\a-\b)$.
The phase portrait and the N-shaped form of the reduced slow manifold of
the Barkley model is illustrated in \fig{bkl}(a), with an example
of a trajectory corresponding to a traveling wave train.

\subsection{The fast subsystem}

Substituting \eq{barkley.pieces}
into equations \eq{zfkfront}, \eq{fhntype.frontspeed} and \eq{fhntype.backspeed}
gives the front velocity 
\begin{gather}
\c=\cfunf(\vf) = \frac{1}{\sqrt{2}}
  \left(1-2\frac{\vf+\b}{\a}\right)>0, \quad
  \vf\in(0,\a/2-\b),                          \eqlabel{barkley.speed.front}
\end{gather}
and the back velocity
\begin{gather}
\c=\cfunb(\vb) = \frac{1}{\sqrt{2}}
  \left(2\frac{\vb+\b}{\a}-1\right)>0, \quad
  \vb\in(\a/2-\b,\a-b).                       \eqlabel{barkley.speed.back}
\end{gather}
As the front and the back have the same speed $\c$, we can obtain
$\vb$ for a given $\vf$ by eliminating $\c$ from system of equations
\eq{barkley.speed.front} and \eq{barkley.speed.back} and resolving it
with respect to 
$\vb$, which gives
\begin{gather}
\vb= \a - 2\b - \vf,                      \eqlabel{mu.3}
\end{gather}
and provides a link for matching with the slow-time problem.
The resulting interval of achievable speeds is 
$\cinterval=(0,(1-2\b/\a)/\sqrt{2})$.

\subsection{The slow subsystem and matching}

Evaluating  expression \eq{fhntype-period} along the stable branches
$\uplu(\v)$ and $\umin(\v)$ of the reduced slow manifold given by
\eq{barkley.pieces}, yields the temporal period of the wave
\begin{gather}
  \eqlabel{barkley.period}
  \BCL = \ln\left[ \frac{(1-\vf)\vb}{(1-\vb)\vf} \right],
\end{gather}

Combining expressions \eq{barkley.speed.front}, \eq{mu.3} and \eq{barkley.period},
finally, yields the CV restitution curve in explicit form
\begin{gather}
\eqlabel{barkley-rc} 
\BCL =  \ln \left( \frac{
    (\a-2\b+\a\c\sqrt2)(2-\a+2\b+\a\c\sqrt2)
  }{
    (\a-2\b-\a\c\sqrt2)(2-\a+2\b-\a\c\sqrt2)
  }\right)            
\end{gather}
which for $\c\in\cinterval$ gives the range $\BCL\in(0,\infty)$.
\Fig{bkl}(b) illustrates this result in comparison with curves
obtained by numerical solution of the full boundary value problem
\eq{full} for equations \eq{fhntype} with right-hand sides given by
\eq{barkley} and periodic boundary conditions as described in section
\ref{BVP}. 

\section{Asymptotic restitution curves  in  the FitzHugh-Nagumo model}
\label{secFHN}

\dblfigure{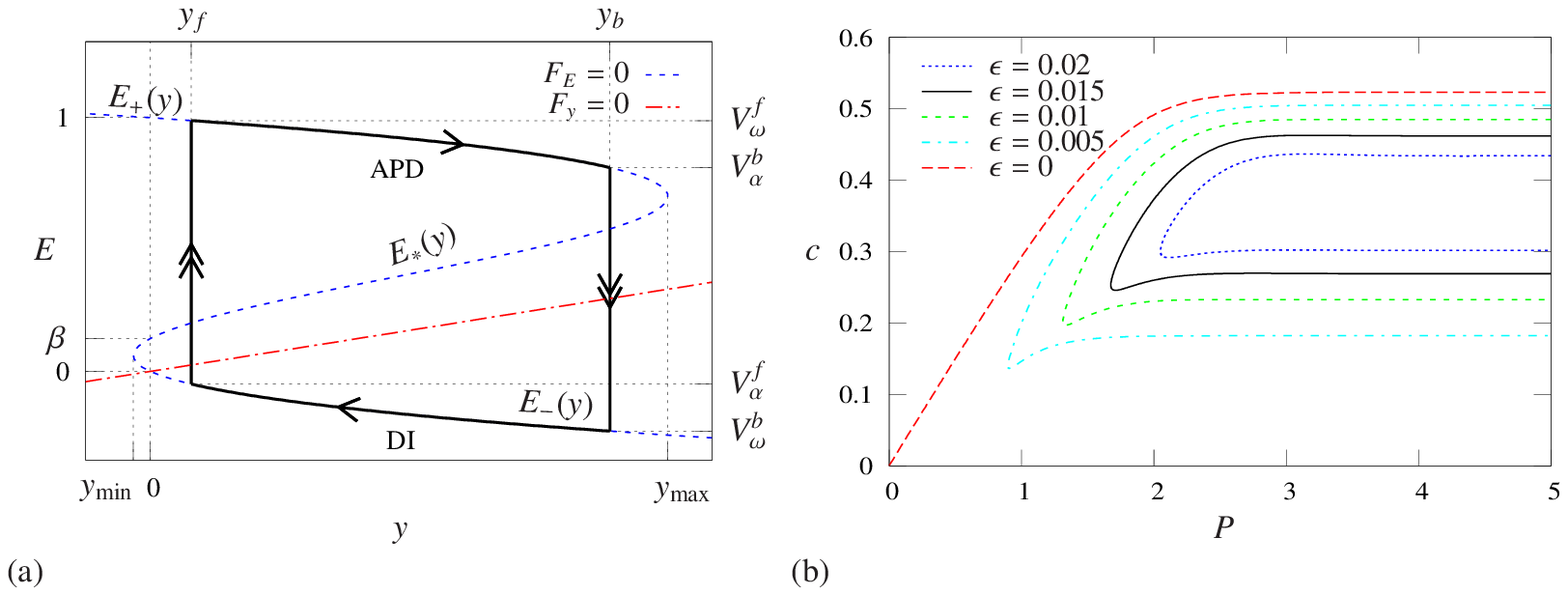}{
  (color online)
  CV asymptotics in the FitzHugh-Nagumo model. 
  Parameters: $\fhnth=0.13$, $\fhnalp=0.37$.
  (a) Schematic of the asymptotic pulse solution in the $(\v,\u)$ plane. 
  (b) CV restitution curves for various values of $\eps$. The curve $\eps=0$
  is the explicit asymptotic result found from
  \eq{fhn.c.front} and \eq{fhn.APD}, \eq{fhn.DI} as detailed in the
  penultimate paragraph of section \ref{secFHN}. 
}{fhn}

\subsection{The model}

We will use the right-hand sides $\p(\cdot,\cdot)$ and $\q(\cdot,\cdot)$ of
the FitzHugh-Nagumo equations in the following form,
\begin{subequations}
\eqlabel{FHN}
\begin{gather}
\eqlabel{FHN1}
\p(\u,\v) = \u\,(1-\u)\,(\u-\fhnth) -\v , \\
\eqlabel{FHN2}
 \q(\u,\v)=\fhnalp \u - \v,
\end{gather}
\end{subequations}
which is related to the original
\cite{FitzHugh-1961,Nagumo-etal-1962} formulation by an affine
transformation of the variables, involving the small parameter $\eps$.
Parameters $\fhnalp$ and $\fhnth$ are assumed to obey
$0<\fhnth<1/2$, $\fhnalp>(1-\fhnth)^2/4$.  
The corresponding phase
portrait is illustrated in \fig{fhn}(a), with a typical
trajectory corresponding to a
traveling wave train.

\subsection{The fast subsystem}
To use the general results on the front velocity  \eq{zfkfront},
we need to know the branches of the reduced slow manifold $\uplumin(\v)$ and
$\uth(\v)$ as functions of the slow variable $\v$. However, unlike the
case of the Barkley model, here this would require using the formula
for the roots of a generic cubic equation. This is rather inconvenient
so we employ an alternative strategy. Given the value of the pre-front
voltage at the lower branch of the reduced slow manifold, $\ualpf=\umin(\vf)$,
we determine from equation \eq{FHN1} the corresponding value of 
slow variable during the front
$\vf=\ualpf(1-\ualpf)\,(\ualpf-\fhnth)$. Further, to find the
corresponding values of $\uth(\vf)$ and the post-front voltage of the 
up-jump $\uomgf=\uplu(\vf)$, we need to solve the cubic
$\u\,(1-\u)\,(\u-\fhnth)-\vf=0$ for which we already know one root,
namely $\u=\ualpf$, so the cubic is divisible by $(\u-\ualpf)$.  Hence
the other two roots are solutions of the resulting quadratic equation, 
which leads to the required values
\begin{subequations}
                                      \eqlabel{fhn.front.pieces}
\begin{gather}
   \uth(\vf)=\frac{1}{2}\left(
     \fhnth+1-\ualpf
     -\sqrt{(\fhnth-1)^2+2\ualpf(\fhnth+1)-3(\ualpf)^2}
   \right),                          \eqlabel{fhn.pieces.uth} \\
   \uomgf=\frac{1}{2}\left(
     \fhnth+1-\ualpf
     +\sqrt{(\fhnth-1)^2+2\ualpf(\fhnth+1)-3(\ualpf)^2}
   \right),                          \eqlabel{fhn.pieces.uplu}
\end{gather}
\end{subequations}
and therefore we get the expression for the front velocity
\begin{gather}
\cfunf(\v)=\frac{\sqrt{2}}{4}\left[
  3\ualpf + 3\sqrt{(\fhnth-1)^2+2\ualpf(\fhnth+1)-3(\ualpf)^2}
  - \fhnth - 1 \right] ,                                     \eqlabel{fhn.c.front}
\end{gather}
assuming $\ualpf=\umin(\v)$,
and a similar expression for the back velocity. Finally, from the
condition $\cfunb(\vb)=\cfunf(\vf)$ we find the
pre-back voltage as
\begin{equation}
\ualpb=\uplu(\vb)=\frac{2}{3}(\fhnth+1)-\ualpf,                       \eqlabel{mu.4}
\end{equation}
which provides the link for matching with the slow-time system.
The interval of bistability required by
assumptions~\ref{A-three-roots} and \ref{A-stability} in this
case is $\vinterval=(\vmin,\vmax)$, where 
$\v_{\max,\min}=(1+\fhnth\pm\fhnsq)(2-\fhnth\mp\fhnsq)(1-2\fhnth\pm\fhnsq)/27$ 
with upper signs are for $\vmax$, lower signs are for $\vmin$, and
$\fhnsq\equiv\sqrt{1-\fhnth+\fhnth^2}$. 
Assumption~\ref{A-direction} narrows this to
$\v\in\vsubinterval=(0,\vmax)$. Assumptions~\ref{A-closed} 
and~\ref{A-monot} for this interval are verified by direct elementary
calculations, giving $\cinterval=(0,\cmax)$ where $\cmax=(1-2\fhnth)/\sqrt2$. 

\subsection{The slow subsystem and matching}

To use the coordinate $\uplumin$ to describe the motion along the reduced slow
manifold, we rewrite \eq{slow.fhntype} as
\begin{gather}
  \Df{\uplumin}{\z} = \Df{\uplumin}{\v} \, \Df{\v}{\z}
  = \frac{
    \fhnalp\uplumin - \uplumin(1-\uplumin)(\uplumin-\fhnth)
  }{
    -3\uplumin^2+2(\fhnth+1)\uplumin-\fhnth
  } .
\end{gather}
Therefore, we have the action potential duration as the time between
$\uomgf$ and $\ualpb$ along the upper branch of the reduced slow manifold,
\begin{gather}
\APD = \int_{\uomgf}^{\ualpb} \frac{
    -3\u^2+2(\fhnth+1)\u-\fhnth
  }{
    \fhnalp\u-\u\,(1-\u)(\u-\fhnth)
  } \,\d{\u} ,                                              \eqlabel{fhn.APD}
\end{gather}
and the diastolic interval as the time between $\uomgb$ and $\ualpf$
along the lower branch of the reduced slow manifold,
\begin{gather}
\DI = \int_{\uomgb}^{\ualpf} \frac{
    -3\u^2+2(\fhnth+1)\u-\fhnth
  }{
    \fhnalp\u-\u\,(1-\u)(\u-\fhnth)
  } \,\d{\u}                                                \eqlabel{fhn.DI}
\end{gather}
and hence the period of the wave $\BCL=\APD+\DI$. 
Note that equations \eq{fhn.APD} and \eq{fhn.DI} have the same
integrand, and only differ in the integration limits, which are
related by relationship~\eq{fhn.pieces.uplu}, a similar expression
relating $\uomgb$ and $\ualpb$, and equation~\eq{mu.4}. 

To summarize,
equation \eq{fhn.c.front} gives the wave velocity $\c$ as the function
of the pre-front voltage $\ualpf$.
Equation \eq{mu.4} gives the pre-back voltage $\ualpb$ as
a function of the pre-front voltage $\ualpf$.
Equation \eq{fhn.pieces.uplu} and its analogue for the back give the
post-front voltage $\uomgf$ and post-back voltage $\uomgb$ as
functions of the pre-front voltage $\ualpf$.
Using those, finally, equations \eq{fhn.APD} and \eq{fhn.DI}
give the wave period $\BCL$, as a function of the pre-front voltage $\ualpf$.
Hence, we have a parametric description of the conduction velocity
restitution curve, $\big(\BCL(\ualpf),\c(\ualpf)\big)$ in a parametric
form with parameter the pre-front voltage $\ualpf$.
The parametric representation can be transformed into explicit
representation by noting that expression \eq{fhn.c.front} is
equivalent to a quadratic equation with respect to the pre-front
voltage $\ualpf$ and can be easily solved to give the desired explicit
expression for $\BCL=\APD+\DI$ as a function of $\c$; the result, however,
is rather lengthy and we omit it here. 

\Fig{fhn}(b) presents a comparison between this
explicit asymptotic $\BCL(\c)$
dependence and the solution of the full periodic boundary-value
problem at various values of $\eps$.

\section{Asymptotic restitution curves in the Caricature Noble model}
\label{asymptotic}

The classical asymptotic theory of slow-fast systems described in the
previous sections is not appropriate for the asymptotic reduction of
cardiac equations which have a different nature, as pointed out in the
Introduction. To develop a fully fledged alternative general
theory is beyond the scope of this paper. Instead, in this section we
study an archetypal ``caricature'' model of cardiac excitation
previously proposed in \cite{Biktashev-etal-2008}. One can think of
this caricature model as a simple example of an ionic cardiac model in which a small
parameter has been embedded so as to reveal explicitly the
non-Tikhonov properties of the equations. The resulting fast and the
slow problems have analytical solutions in closed form which makes the
model convenient for investigation of well-posedness of the asymptotic
reduction of the CV restitution problem in this particular case.

\subsection{The model}

We consider the following set of equations
\cite{Biktashev-etal-2008},
\begin{subequations}
\eqlabel{caric}
\begin{align}
\eqlabel{caric.E}
  &\df{\E}{\t} =
   \frac{1}{\eps}\GNa\,(\ENa-\E)\,\Heav(\E-\East)\,\h 
    + \tildegtwo(\E)\,\n^4+\Gtilde(\E)
    + \eps \ddf{\E}{\x}, \\
\eqlabel{caric.h}
  &\df{\h}{\t} =
  \frac{1}{\eps}\Fone\,\big(\Heav(\Edagger-\E)-\h\big),\\
\eqlabel{caric.n}
  &\df{\n}{\t} = \Ftwo\,\big(\Heav(\E-\Edagger)-\n\big),
\end{align}
\end{subequations}
where
\begin{gather}
  \tildegtwo(\E) = \Gtwol\Heav(\Edagger-\E)+\Gtwor\Heav(\E-\Edagger),\nonumber\\
\Gtwol=-2,  \quad
\Gtwor=-9, \nonumber\\
  \Gtilde(\E) = \left\{\begin{array}{lll}
  \kone (\Eone-\E), &  \E\in(-\infty,\Edagger), \\
  \ktwo (\E-\Etwo), &  \E\in[\Edagger,\East), \\
  \kthree (\Ethree-\E), &  \E\in[\East,+\infty),
  \end{array}\right. \nonumber\\
\kone=3/40, \quad 
\ktwo=1/25, \quad 
\kthree=1/10, \quad 
\Eone=-280/3,                                          \nonumber\\
\Etwo=(\kone/\ktwo+1)\Edagger-\Eone\kone/\ktwo = -55,  \nonumber\\
\Ethree=(\ktwo/\kthree+1)\East-\Etwo\ktwo/\kthree = 1, \nonumber\\
\Fone=1/2, \quad
\Ftwo=1/270,                                           \nonumber\\
\ENa=40, \quad
\Edagger=-80, \quad
\East=-15, \quad
\GNa=100/3,                                            \eqlabel{caricfunctions}
\end{gather}
and where $\Heav(\cdot)$ is the Heaviside unit step function.

  The time in this model is measured in $\ms$ and the voltages $\E$,
  $\Eone$, $\Etwo$, $\Ethree$, $\East$, $\Edagger$ are measured in
  $\mV$. Correspondingly, the units of $\tildegtwo$,
  $\Gtwol$, $\Gtwor$ are $\mV/\ms$, and the units of
  $\GNa$, $\Fone$ and $\Ftwo$ are $\ms^{-1}$.  As discussed above
  in section~\ref{BVP}, the space scale is chosen to get the
  convenient value $\eps$ for the coefficient at the voltage diffusion
  term, so the dimensionality of $x$ in \eq{caric} is given by the
  ``space unit'' $\su=\ms^{1/2}$. The real physical lengths are given
  by $\x\D^{1/2}$ where $\D$ is the tissue voltage diffusion
  coefficient in the direction of wave propagation.  The rest of the
  quantities in \eq{caric} are dimensionless.

This system is obtained from the authentic Noble model of Purkinje
fibers \cite{Noble-1962} using a set of verifiable assumptions and well
defined simplifications as detailed in
\cite{Biktashev-etal-2008}. The main features of equations \eq{caric}
which make them an appropriate illustration are:
\begin{enumerate}[(a)]
\item They reproduce exactly the asymptotic structure of the
  authentic Noble model \cite{Noble-1962}, which is guaranteed by the
  embedding of the artificial small parameter $0<\eps\ll 1$. The
  authentic Noble model is the prototype of all contemporary
  voltage-gated cardiac models, and we believe that the asymptotic
  structure of \eq{caric} is rather generic in this
  class. Realistic voltage-gated cardiac models do
  not have explicit small parameters already present in them; or,
    rather, they have so many parameters that it is not a
    straightforward task which of them to use for asymptotics.  Hence
    we employ a procedure of embedding artificial small parameters, as
    discussed \eg\ in
  \cite{Biktashev-etal-2008}. An example of the embedding procedure
  appears in section \ref{applicationBR} below, where the
  Beeler-Reuter model \cite{Beeler-Reuter-1977} is discussed.
\item Equations \eq{caric}  have the simplest
possible functional form consistent with property (a). Most
functions in the right-hand side are replaced by constants as justified
in \cite{Biktashev-etal-2008} which allows analytical solutions to be
obtained in closed form. This in turn makes it possible to prove
the well-posedness of the asymptotic boundary value problem to be
formulated below.
\end{enumerate}

For brevity, we shall call this model ``Caricature Noble''.

\subsection{The asymptotic reduction of the CV restitution
  boundary-value problem} 
\label{AsymptFormulation}

Model \eq{caric} contains an explicit small parameter $\eps$ embedded
in essentially the same way as it would be in a realistic model. In
this section we demonstrate how this may be used for simplification of
the Caricature Noble model or, indeed, of a more realistic ionic model.

A slow-time subsystem which describes the plateau and the recovery
stages can be obtained immediately from equations \eq{caric}, by
taking the limit $\eps \to 0$. At time scales much longer than $\eps$,
the second equation implies $\h\to\Heav(\Edagger-\E)$. Hence the first
term of equation \eq{caric.E} is proportional to
$\Heav(\E-\East)\Heav(\Edagger-\E)=0$ which vanishes in the limit
$\eps\to0$ despite the large factor $\eps^{-1}$ in front of 
it\footnote{
  This is an attempt to summarize briefly the essence of the
  non-Tikhonov asymptotics of this and similar
  models.  For a more   detailed treatment, see our previous
  publications,   \eg~\cite{Biktashev-etal-2008}.
}.
The diffusion term $\@_{\x}^2\E$ vanishes in the same limit and
we are left with  the slow-time system,
\begin{subequations}
\eqlabel{caric-slow}
\begin{gather}
\eqlabel{slow.E}
  \Df{\E}{\zt} = \tildegtwo(\E)\,\n^4+\Gtilde(\E),     \\
\eqlabel{slow.n}
  \Df{\n}{\zt} = \Ftwo\,\big(\Heav(\E-\Edagger)-\n\big) ,
\end{gather}
\end{subequations}
where $\zt=\t-\x/\c$ is the traveling wave coordinate, 
which we use in this section instead of our standard choice
of $\z=\x-\c\t$. 
A fast-time subsystem of equations \eq{caric} can be obtained by
stretching time and space, $\T = \t/\eps$, $\X=\x/\eps$, taking the
limit $\eps\to0$ and neglecting the equation for $\n$ which decouples
from the rest. It is useful to distinguish explicitly the functions of the old
from the functions of the new independent variables,
say
$\E(\x,\t)=\Efast(\X,\T)=\Efast(\x/\eps,\t/\eps)$
and
$\h(\x,\t)=\hfast(\X,\T)=\hfast(\x/\eps,\t/\eps)$.
It is also useful to introduce
at this stage the following non-dimensionalization (which, as noted
  above, is different from other sections and specific for this particular model) 
\begin{gather}
  \V=\frac{\Efast-\East}{\ENa-\East}, \;
  \xa = \X \sqrt{\Fone},\;
  \ta=\Fone \T, \nonumber\\
  \GNaa = \frac{\GNa}{\Fone}, \;
  \C=\c/\sqrt{\Fone} .         \eqlabel{nondim}
\end{gather}
In these variables, the travelling wave ansatz becomes $\za=\ta-\xa/\C$.
As a result of these transformations we
obtain the following fast-time model of the wave front,
\begin{subequations}
\eqlabel{caric-fast}
\begin{gather}
\Df{\V}{\za} = \frac{1}{\C^2}\frac{\d^2 \V}{\d \za^2} +
\GNaa\,(1-\V)\,\Heav(\V)\,\hfast,
\eqlabel{fast.v}\\
\Df{\hfast}{\za}= \Heav(\Vdagger-\V)-\hfast, \eqlabel{fast.h}
\end{gather}
\end{subequations}
where $\Vdagger=(\Edagger-\East)/(\ENa-\East)<0$.
In a periodic wave train, a front propagates in the tail of the
preceding wave, so slow pieces described by \eq{caric-slow} and
fast pieces described by \eq{caric-fast} alternate, and there is one
slow piece and one fast piece per period,   as opposed to two fast
pieces and two slow pieces in classical Barkley and FitzHugh-Nagumo
models. The matching points for the van Dyke rule are: (a) 
the end of a slow piece $\zt=\BCL$ corresponds to the beginning of the
fast piece $\za\to-\infty$ and (b) the end of the fast piece
$\za\to\infty$ corresponds to the beginning of the next slow piece
$\zt=0$.
%
%
This situation is summarized by the following set of boundary
conditions
\begin{subequations}
 \eqlabel{asympBCs}
\begin{gather}
  \V(-\infty)=\Valpha, \ \
  \V(\infty)=\Vomega, \ \
  \left.\Df{V}{\za}\right|_{\za\to\infty}=0, \nonumber\\
  \hfast(-\infty)=1, \ \
  \V(0)=0, \eqlabel{asympBCs-fast}
\end{gather}
\end{subequations}
together with
\addtocounter{equation}{-1}
\begin{subequations}
\addtocounter{equation}{1}
\begin{gather}
  \E(0)=\East+\Vomega\,(\ENa-\East), \nonumber\\
  \E(\BCL)=\East+\Valpha\,(\ENa-\East), \\
  \n(0)=\n(\BCL), \nonumber
\eqlabel{asympBCs-slow}
\end{gather}
\end{subequations}
where $\C\in(0,\infty)$, $\BCL\in(0,\infty)$,
$\Valpha\in(-\infty,\Vdagger)$ and $\Vomega\in(0,1)$
are parameters to be found.
Condition $\V(0)=0$ is a pinning condition as discussed above in
\eq{full}, \ie\ we choose $\Ev=\East$.
Condition $\hfast(-\infty)=1$ follows from
matching condition $\hfast(-\infty)=\h(\BCL)$,
by noting that $\h$
in the slow time system is given by $\h=\Heav(\Edagger-\E)$
and $\E(\BCL)<\Edagger$.
The asymmetry of the conditions imposed
at $\za\to\pm\infty$ can be understood
by analysing the $\za\to\pm\infty$ asymptotics of the linearized
problem: the condition on the $\za$-derivative, $\V'(-\infty)=0$,
and the condition $\hfast(\infty)=0$ are
satisfied automatically for open sets of solutions, whereas
the the condition on the $\za$-derivative,
$\V'(+\infty)=0$, excludes solutions exponentially growing as
$\za\to+\infty$ and the condition $\hfast(-\infty)=1$ excludes solutions
exponentially growing as $\za\to-\infty$.

Equations \eq{caric-slow} and \eq{caric-fast} together with the
boundary conditions \eq{asympBCs} form a set of coupled boundary
value problems representing an asymptotic description of CV
restitution. The slow system is of order 2 while the fast
system is of order 3 and there are 4 unknown constants namely $\C$, $\BCL$,
$\Valpha$ and $\Vomega$. Hence 9 conditions are needed to select a unique solution,
while \eq{asympBCs} provide only 8 conditions. Hence a one-parameter
family of solutions may be found where the wave velocity is a
function of the wave period, $\C(\BCL)$.

The asymptotic boundary-value problem
\eq{caric-slow},\eq{caric-fast} and \eq{asympBCs} of CV restitution is
essentially simpler than the full one \eq{full} for equations
\eq{caric}. Indeed, the small parameter has been eliminated and the
resulting system is no longer stiff. Furthermore, the right-hand sides
of equations are simpler and each stage of the action potential is
modeled asymptotically by a system of lower dimension. However, to be
useful the coupled asymptotic boundary value problem must satisfy two
essential requirements: (a) the coupled problems must be well-posed
(b) their asymptotic solution must provide a good approximation to the
solution of the full non-asymptotic problem. It is not obvious that
the asymptotic formulation of the CV restitution problem satisfies
either of these requirements in the non-Tikhonov case under consideration.
While a proof of properties (a) and (b) in the case of any arbitrary
voltage-gated cardiac model is beyond the scope of this paper, in this
section we prove the well-posedness of the
archetypal Caricature Noble problem \eq{caric-slow}, \eq{caric-fast} and
\eq{asympBCs}. The convergence of the asymptotic and full solutions is
demonstrated numerically.

\subsection{The fast subsystem}

\subsubsection{Exact solution}

To solve the fast-time equations \eq{caric-fast} and
\eq{asympBCs-fast}, we follow the ideas presented in
\cite{Hinch-2002,Simitev-Biktashev-2006}. Since the
right-hand-side  of equation \eq{fast.v} is a piece-wise function of
voltage, we distinguish three intervals in terms of voltage
separated by $\Vdagger$ and 0, or alternatively in terms of
the wave coordinate $\za$ we use the intervals
$\za\in(-\infty,\zadagger]$,
$\za\in[\zadagger,0]$ and $\za\in[0,\infty)$ with internal
boundaries $\za=\zadagger$ and $\za=0$ for which the equations
$\V(\zadagger)=\Vdagger$ and $\V(0)=0$ are satisfied, and impose 
natural continuity conditions at the internal boundaries.
Exact analytical solution of the fast system can be obtained
by first solving the $\hfast$-equation \eq{fast.h} which is separable
and independent of $\V$ and then substituting its solution in the
voltage equation \eq{fast.v}. In the first two intervals,
$\za\in(-\infty,\zadagger]$ and $\za\in[\zadagger,0]$, equation \eq{fast.v}
is then readily solvable. The internal boundary point
$\zadagger$ can be obtained by matching the solutions in these two
intervals. To solve the voltage equation \eq{fast.v} in the third
interval $\za\in[0,\infty)$ we use the auxiliary change of variables,
\begin{gather}
\s=2\C\sqrt{\GNaa} \exp \big((\zadagger-\za)/2\big), \nonumber\\
\w(\s)=\big(\V(\za)-1\big)\exp\left(-\C^2\za/2\right),
\label{change-to-bessel}
\end{gather}
and obtain a modified Bessel equation of order $\C^2$,
\begin{gather}
\s^2 \frac{\d^2\w}{\d\s^2} + \s\frac{\d\w}{\d\s} - (\C^4+\s^2)\w = 0,
\end{gather}
the solutions of which are a linear superposition of the modified Bessel functions
$\BesselI_{\C^2}(\s)$ and $\BesselK_{\C^2}(\s)$ of order $\C^2$ 
\cite{Abramowitz-Stegun-1965}. The requirement of boundedness of the 
solution at infinity eliminates the $\BesselK_{\C^2}(\s)$ term. The 
value of the post-front voltage $\Vomega$ is obtained as the limit of
the expression for the voltage as $\za\to\infty$ and using formula
\cite[(9.6.7)]{Abramowitz-Stegun-1965}. In summary, the exact 
analytical solution of equations \eq{caric-fast} is
\footnote{
  At $\Vdagger=0$, $\zadagger=0$,
  this solution coincides with that of Hinch~\cite{Hinch-2002} at $v_{r,\textrm{eff}}=0$.
}
\begin{subequations}
\eqlabel{fastexact}
\begin{gather}
\eqlabel{fastexact.n}
\hfast(\za) = \begin{cases}
     1,  &                            \za\in(-\infty,\zadagger]\\
     \exp\left(\zadagger-\za\right), &      \za\in[\zadagger,\infty)
\end{cases}
\end{gather}
\end{subequations}
\addtocounter{equation}{-1}
\begin{subnumcases}{\V(\za) = }
\addtocounter{equation}{1}
\eqlabel{fastexact.v12}
  (\Vdagger-\Valpha)\exp\big( \C^2(\za-\zadagger)\big)+ \Valpha, & \nonumber\\
\hspace{3cm} \text{for }   \za\in(-\infty,0]  & \\
\eqlabel{fastexact.v3}
  1- \exp(\C^2\za/2)
  \frac{\displaystyle
    \BesselI_{\C^2}\Big(2\C\sqrt{\GNaa}\exp\big((\zadagger-\za)/2\big)\Big)}{\displaystyle \BesselI_{\C^2}\Big(2\C\sqrt{\GNaa}\exp(\zadagger/2)\Big)}, \nonumber\\
  \hspace{3cm} \text{for }   \za\in[0,\infty) &
\end{subnumcases}
where the internal boundary point $\zadagger$ is given by
\begin{gather}
\zadagger=\frac{1}{\C^2} \ln \left( 1-\frac{\Vdagger}{\Valpha} \right) < 0
\end{gather}
and the post-front (peak) voltage is
\begin{gather}
\Vomega=\lim\limits_{\za\to+\infty}\V(\za)=
1-\frac{\Big(\so/2\Big)^{\C^2}}{\Gamma(\C^2+1)\,\BesselI_{\C^2}\Big(\so\Big)},
                               \eqlabel{Vomega}
\end{gather}
where $\Gamma(\cdot)$ is the Gamma function
\cite{Abramowitz-Stegun-1965}. The dispersion relation is then found from
the continuity of the derivative of voltage $\V$ at $\za=0$ and 
takes the form
\begin{gather}
\eqlabel{disprel}
\C^2\,\left(1 - 2\Valpha \right) = \frac{\so
  \BesselI'_{\C^2}(\so)}{\BesselI_{\C^2}(\so)}, \ \ \ \so
=2\C\sqrt{\GNaa}\left(1-\frac{\Vdagger}{\Valpha}\right)^{{1}/{(2\C^2)}}
\end{gather}
where the prime indicates a derivative with respect to the argument
$\s$. At a given pre-front voltage $\Valpha$, the wave velocity \
$\C$ of the travelling impulses can be found
as a solution of the dispersion relation \eq{disprel} and we remind
that for comparison with numerical results the wave velocity should be
transformed back to the original variables,
$\c= \C\sqrt{\Fone}$.  Once again, note that the pre-front
voltage $\Valpha$ cannot be found from conditions \eq{asympBCs-fast}
alone and so the entire front solution is a one-parameter function as
expected. 

The existence of solutions to the fast-time boundary value
problem thus ultimately depends on the existence of solutions to
the transcendental dispersion relation \eq{disprel}. Solutions are
guaranteed by the following
\begin{prop} \label{prop-one}
For every set of parameters such that $\Valpha<\Vdagger<0$ and $\C>0$,
there exists a unique value of the excitability parameter
$\GNaa=\HH(\C,\Valpha,\Vdagger)>0$
which solves \eq{disprel}.
\end{prop}

\begin{proof}
$1^{\circ}$ We will need an important property of modified Bessel functions: the ratio
\[
  \RR_{\ggam}(\pp)=\BesselI_{\ggam+1}(\pp)/\BesselI_{\ggam}(\pp)
\]
is a strictly increasing function of its argument
$\pp\in(0,\infty)$ for any order $\ggam>0$. This follows directly
from the estimate $\RR'_{\ggam}(\pp)>0$,  \cite[p. 243]{Amos-1974}.

$2^\circ$
Moreover, it is easily established from the asymptotics of
Bessel functions that
$\lim\limits_{\pp\to0}\RR_{\ggam}(\pp)=0$
and
$\lim\limits_{\pp\to\infty}\RR_{\ggam}(\pp)=1$.

$3^{\circ}$
The right-hand side of \eq{disprel} is a composite of the function
\[
  \FF_{\ggam}: \pp \mapsto \pp \BesselI'_{\ggam}(\pp)/\BesselI_{\ggam}(\pp),
\]
further depending on $\ggam=\C^2$ as a parameter, and the function
\[
  \SS : \GNaa \mapsto \so=2\C\sqrt{\GNaa}\left(1-\Vdagger/\Valpha\right)^{\frac{1}{2\C^2}}
\]
further depending on $\C$, $\Valpha$ and $\Vdagger$ as 
parameters. In the assumptions made, $\SS$ is obviously strictly
increasing as a function of $\GNaa$, and maps 
$(0,\infty)\to(0,\infty)$.
Using the recurrence relations \cite[(9.6.26)]{Abramowitz-Stegun-1965}, we
can rewrite the definition of the function $\FF_{\ggam}$ as
\begin{equation}
  \FF_{\ggam}(\pp) = \ggam + \pp \RR_{\ggam}(\pp).  \eqlabel{FF-vs-RR}
\end{equation}
By $1^\circ$ and $2^\circ$, we have that $\FF_{\ggam}(\pp)$ is
strictly increasing (and therefore invertible) and maps $(0,\infty)\to(\ggam,\infty)$. Overall,
we conclude that the right-hand side of \eq{disprel} is a strictly
increasing function of $\GNaa$ defined for all $\GNaa>0$ and with the
range of $(\ggam,\infty)$.

$4^\circ$
The left-hand side of \eq{disprel} does not depend on the 
parameter $\ggam$ and, since the pre-front voltage $\Valpha<0$ 
by assumption, it lies within $(\ggam,\infty)$ which by $3^\circ$
is the range of the right hand side. Hence a solution $\GNaa>0$ always exists.
Moreover, since by $3^\circ$ the right-hand side is strictly
monotonic, the solution is unique. 
\end{proof}

Denoting by $\GG_{\ggam}$ the inverse function to $\FF_{\ggam}$
at the constant order $\ggam$, the existence of which has 
just been established in $3^\circ$ above, the solution of
\eq{disprel} can be written as
\begin{equation}
  \GNaa = \HH(\C,\Valpha,\Vdagger) \equiv \frac{1}{8\C}
    \left(1-\frac{\Vdagger}{\Valpha}\right)^{-1/(2\C^2)}
    \GG_{\C^2}\left((1-2\Valpha)\C^2\right).
    \eqlabel{levelset}
\end{equation}

\subsubsection{Bounds and asymptotics}

We have demonstrated the existence of solutions to the fast 
subsystem \eq{caric-fast} and \eq{asympBCs-fast}. We will now
consider some estimates related to this solution which will lead
to convenient explicit approximations of the propagation
speed and the minimal excitability required for wave propagation.
We treat $\Vdagger\in(-\infty,0)$ as a parameter
characterising the system (and omit dependence on it in the function
notations), while $\C$ and $\Valpha$ as variables characterising a
particular front solution.

\paragraph{Lower bounds on the  excitability $\GNaa$} From 
$1^\circ$ and $2^\circ$ we know that $\RR_\ggam(\pp)<1$, where the
inequality becomes approximate equality for large $\pp$. We shall
denote this as $\RR_\ggam(\pp)\lesssim 1$. With account of
\eq{FF-vs-RR}, this implies that $\FF_{\ggam}(\pp)\lesssim\ggam+\pp$.
A sharper upper bound,
$\RR_\ggam(\pp)\lesssim((\pp^2+\ggam^2)^{1/2}-\ggam)/\pp$, is given in
\cite[(9)]{Amos-1974} and implies
\begin{equation}
    \FF_{\ggam}(\pp) \lesssim  (\ggam^2+\pp^2)^{1/2},
  \qquad \ggam>0, \; \pp>0.
  \eqlabel{FF-upper}
\end{equation}
Substituting this into \eq{disprel}, we get the more easily
tractable approximation
\begin{equation}
  \C^2(1-2\Valpha) \lesssim  \left(
    4\C^2\GNaa\left(1-\Vdagger/\Valpha\right)^{1/\C^2} + \C^4
  \right)^{1/2} .
  \eqlabel{upper-inequality}
\end{equation}
Resolving this with respect to the excitability parameter 
$\GNaa$, we get a lower bound for it in the form
\begin{equation}
  \GNaa \gtrsim \HHlwr(\C,\Valpha) \equiv \Valpha(\Valpha-1) \C^2 \left(
    \frac{-\Valpha}{\Vdagger-\Valpha}
  \right)^{1/\C^2} .                           
\eqlabel{ggg-lower}
\end{equation}
The minimum of $\HHlwr(\C,\Valpha)$ defined in \eq{ggg-lower} with
respect to the wave speed $\C\in(0,\infty)$ at a constant
pre-front voltage $\Valpha\in(-\infty,\Vdagger)$ is achieved for
\begin{equation}
  \C=\C^*(\Valpha)=\left(
    \ln\left(\frac{-\Valpha}{\Vdagger-\Valpha}\right) \right)^{1/2}
  \eqlabel{C-ast}
\end{equation}
and is equal to
\begin{equation}
  \HHlwr^*(\Valpha) = e \Valpha(\Valpha-1)
  \ln\left(\frac{-\Valpha}{\Vdagger-\Valpha}\right).
  \eqlabel{HHlwr-ast}
\end{equation}
Hence, for any fixed value of the pre-front voltage
$\Valpha\in(-\infty,\Vdagger)$ and the excitability parameter
$\GNaa>\HHlwr^*(\Valpha)$, there exist two solutions for the wave
speed $\C$, namely 
$\Cfast(\Valpha,\GNaa)>\C^*(\Valpha)$ and
$\Cslow(\Valpha,\GNaa)<\C^*(\Valpha)$.

Similarly, the minimum of $\HHlwr(\C,\cdot)$ with respect to
the pre-front voltage $\Valpha\in(-\infty,\Vdagger)$ for a fixed
value of the wave speed $\C\in(0,\infty)$ is defined by
\begin{gather}
  \Valpha = \Valpha^{\star}(\C) = - \frac{ \auxvar +
    \sqrt{\auxvar^2-2\C^2(\C^2+1)\Vdagger}}{4\C^2}, \\
  \auxvar=\Vdagger-2\C^2\Vdagger-\C^2, \nonumber
\end{gather}
which is the negative root of the quadratic equation
\begin{equation}
   2\C^2 \Valpha^2 - (\Vdagger + \C^ 2+ 2\C^2\Vdagger)\Valpha +
   \Vdagger + \C^2 \Vdagger = 0 ,
   \eqlabel{star-quadratic}
\end{equation}
and a corresponding value
$\HHlwr^{\star}(\C)=\HHlwr(\C,\Valpha^{\star}(\C))$
(the explicit equation for which is lengthy and of no further
consequence so we omit it). Hence for any given wave speed $\C$,
and excitability parameter $\GNaa>\HHlwr^{\star}(\C)$, we have
two solutions for the pre-front voltage $\Valpha$, separated by
the value $\Valpha^{\star}(\C)$.

Alternatively, equation \eq{star-quadratic} can be resolved with
respect to $\C$ to obtain
\begin{equation}
  \C = \C^{\star}(\Valpha) = \left( \frac{
      -\Vdagger(1-\Valpha)
  }{
    (1-2\Valpha)(\Vdagger-\Valpha)
  } \right)^{1/2} .
\end{equation}
The absolute minimum of $\HHlwr(\C,\Valpha)$ over
$\{(\C,\Valpha)\}=(0,\infty)\times(-\infty,\Vdagger)$
is achieved when
$\C^*(\Valpha)=\C^{\star}(\Valpha)$
or, equivalently,
\[
  \ff(\Valpha) \equiv \left(\frac{\C^*}{\C^\star}\right)^2 \equiv
  \frac{(\Vdagger-\Valpha)(1-2\Valpha)}{-\Vdagger(1-\Valpha)}\ln\left(\frac{-\Valpha}{\Vdagger-\Valpha}\right)=1.
\]
The left-hand side $\ff(\Valpha)$ of this equation is a strictly decreasing function
in $\Valpha\in(-\infty,\Vdagger)$, which is established \eg\ by differentiation and using
the estimate $\ln(x)>1-1/x$ for $x>1$ (the calculations are tedious
but elementary and we omit them). Besides,
$\lim\limits_{\Valpha\to-\infty}\ff(\Valpha)=2$ and
$\lim\limits_{\Valpha\to\Vdagger-0}\ff(\Valpha,\Vdagger)=0$,
hence the above equation  for the minimum always has a unique
solution, $\Valpha=\Valpha^{*\star}$, with corresponding
$\C=\C^{*\star}$ and 
$\HHlwr=\HHlwr^{*\star}$,
all depending on $\Vdagger\in(-\infty,0)$ as a parameter.

To summarize, we have established that the lower bound $\HHlwr(\C,\Valpha)$
given by \eq{ggg-lower} has a unique absolute minimum
$\HHlwr^{*\star}=\HHlwr(\C^{*\star},\Valpha^{*\star})$ in
$\{(\C,\Valpha)\}=(0,\infty)\times(-\infty,\Vdagger$), tends to
infinity as $\C\to0$ and $\C\to\infty$ uniformly in
$\Valpha\in(-\infty,\Vdagger)$ and as $\Valpha\to-\infty$ and
$\Valpha\to\Vdagger$ uniformly in $\C\in(0,\infty)$. Hence by
\cite[Theorem 3.1]{Milnor-etal-1963},
all level sets
$\HHlwr(\C,\Valpha)=\const>\HHlwr^{*\star}$ are diffeomorphic
to each other, and thus are simple closed curves circumventing
$(\C^{*\star},\Valpha^{*\star})$. Moreover, for any
$\GNaa>\HHlwr^{*\star}$, there are two values of $\Valpha=\Valpha_{-}$
and $\Valpha=\Valpha_{+}>\Valpha_{-}$, that are exactly two solutions of
$\HHlwr^*(\Valpha)=\GNaa$, such that for any
$\Valpha\in(\Valpha_{-},\Valpha_{+})$, there exist exactly two
solutions for the front velocity, $\C=\Cfast(\Valpha,\GNaa)$ and
 $\C=\Cslow(\Valpha,\GNaa)$, where
 $\Cslow(\Valpha,\GNaa)<\C^*(\Valpha)<\Cfast(\Valpha,\GNaa)$.

Consequently, the level sets $\HH(\C,\Valpha)=\GNaa$, which are
the dispersion curves defined by \eq{disprel}, are located wholly
inside corresponding level sets of the lower bound
$\HHlwr(\C,\Valpha)=\GNaa$. In particular, there are no solutions for
\eq{disprel} for $\GNaa<\HH^{*\star}$.

\paragraph{Estimates of  the propagation velocity}
Inequality \eq{upper-inequality} can also be explicitly resolved with
respect to $\C$, giving double-sided inequality
\begin{equation}
  \Clwr \lesssim \C \lesssim \Cupr           \eqlabel{C-double}
\end{equation}
where
\begin{equation}
  \Clwr =
  \left( \frac{
    \ln\left(1 - \Vdagger/\Valpha\right)
  }{
    \LambertW_{-1}\left(\frac{\Valpha(\Valpha-1)}{\GNaa}\ln\left(1-\Vdagger/\Valpha\right)\right)
  } \right)^{1/2},                          \eqlabel{Clower}
\end{equation}
\begin{equation}
   \Cupr = \left( \frac{
    \ln\left(1 - \Vdagger/\Valpha\right)
  }{
    \LambertW_0\left(\frac{\Valpha(\Valpha-1)}{\GNaa}\ln\left(1-\Vdagger/\Valpha\right)\right)
  } \right)^{1/2},                          \eqlabel{Cupper}
\end{equation}
and $\LambertW_0(\cdot)$  and $\LambertW_{-1}(\cdot)$ are the principal
and the alternate branch of the Lambert
function \cite{Corless-etal-1996}, respectively.

The upper inequality in estimate \eq{C-double} becomes
an asymptotic equality for the faster velocity
in the limit of rescaled coordinate $\so\to\infty$ which is
achieved \eg\ for large excitability $\GNaa$ and wave speed
$\C$ at fixed pre-front voltage $\Valpha$.
In this limit, the arguments of the
Lambert functions are small, and since $\LambertW_0(z)=z+\O{z^2}$, $z\to0$, we have
\begin{equation}
  \Cfast \approx \Cupr \approx \left(\frac{\GNaa}{\Valpha(\Valpha-1)}\right)^{1/2} ,
  \qquad \GNaa\to\infty
  \eqlabel{Climfast}
\end{equation}
which agrees with \cite[(39)]{Hinch-2002}, as should be expected since
in this limit the difference between our model and \cite{Hinch-2002}
is inessential.

Similarly, using a crude estimate for the alternate branch of
Lambert function $\LambertW_{-1}(z) = \ln(-z) \left( 1 + \o{1}
\right)$ for $z\to-0$, see \cite{Corless-etal-1996}, we find for the lower
inequality in \eq{C-double}
\[
  \Cslow \approx
  \Clwr \approx \left( \frac{
    \ln\left(1-\Vdagger/\Valpha\right)
  }{
    \ln\left( - \frac{\Valpha(\Valpha-1)}{\GNaa} \ln\left(1-\Vdagger/\Valpha\right)\right)
  }\right)^{1/2},
  \qquad \GNaa\to\infty.
\]
However, because the convergent series representation of the
Lambert function $\LambertW_{-1}$ is in terms of logarithms, it
converges rather slowly, and thus the asymptotic above does not 
give a good approximation for the standard parameter values. A better,
more accurate and simpler asymptotic can be obtained directly from
\eq{disprel} in the limit $\so\to0$ and using
$\BesselI'_\nu(z)/\BesselI_{\nu}(z)=\nu/z+z/(2\nu+2)+\O{z^3}$,
$z\to0$, which can be easily obtained from 
\cite[(9.6.7)]{Abramowitz-Stegun-1965}. This leads to
\begin{equation}
\Cslow \approx \left(\frac{\ln(1-\Vdagger/\Valpha)}
    {\ln(-\Valpha/\GNaa)}\right)^{\frac{1}{2}},
  \qquad \GNaa\to\infty.
\eqlabel{Climslowas}
\end{equation}
Note that since the estimate \eq{Climslowas} is only asymptotic
rather than uniform, it
is not necessarily defined for all $\GNaa$ for which the solution
exists. Indeed, since $\ln(1-\Vdagger/\Valpha) < 0 $ for all $\Valpha<\Vdagger$, we
must require that $\ln(-\Valpha/\GNaa) < 0$ in order that
$\C\in\Real$. This is only possible when $\GNaa>-\Valpha$,
although according to \eq{HHlwr-ast}, for $\Vdagger\to0$, the minimal
excitability $\HHlwr^*\to0$.

\dblfigure{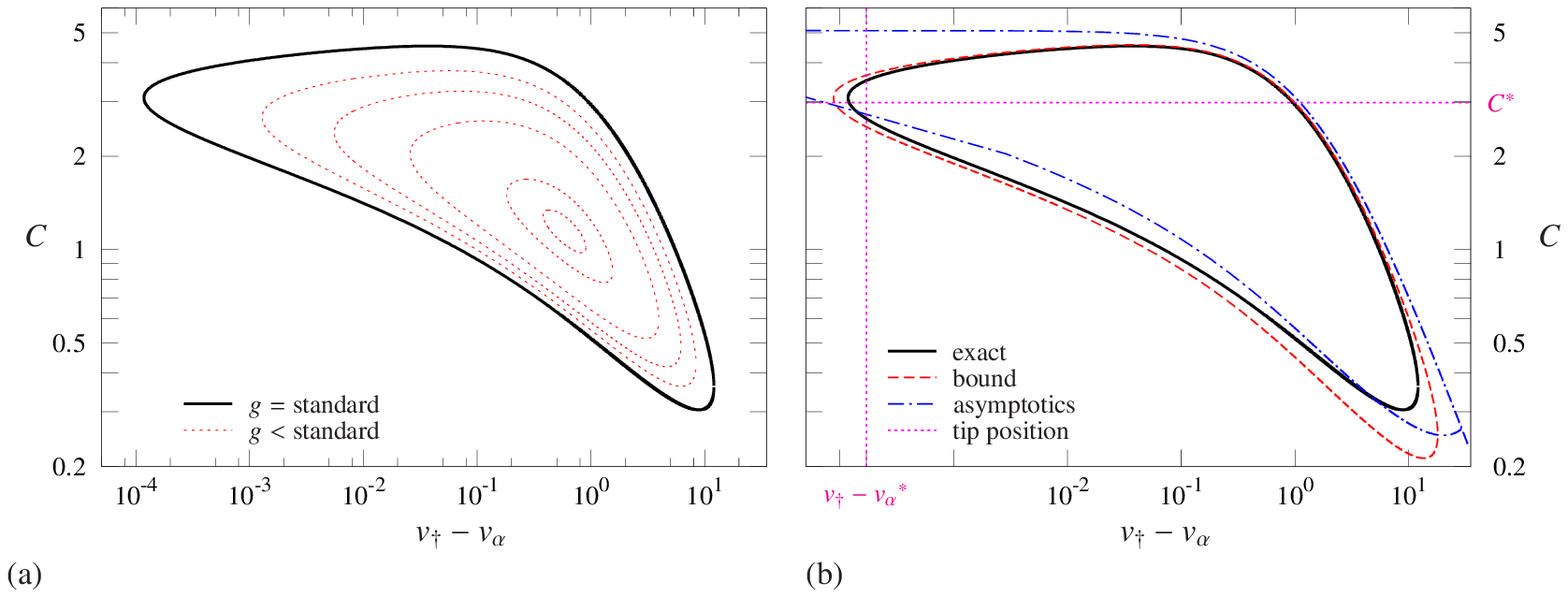}{
  (color online)
  Solutions of the dispersion relation \eq{disprel}. 
  (a) Accurate numerical solutions, for the standard value of
  $\GNaa=200/3$ (thick solid black line) and smaller
  values of $\GNaa=18, 20, 30, 40, 50$ (thin dashed red lines,
  $\GNaa$ increasing from inside out) and the standard value of $\Vdagger$. 
  (b) The accurate numerical solution for all standard parameters (solid black line)
  against the solution provided by the lower bound $\HHlwr(\C,\Valpha)$ \eqlabel{g-lower}
  (dashed red line) and the asymptotics 
  \eq{Climfast} for the fast branch and
  \eq{Climslowas} for the slow branch (dash-dotted blue line). 
  The magenta dotted lines indicate the position of the tip of the
  curve $\Valpha^*$ as estimated by \eq{tip}. 
}{valc}

\subsubsection{On the properties of the exact solution}
The above estimates of the propagation velocity $\C$ are
derived from the lower bound of the excitability
$\HHlwr(\C,\Valpha)$.  The exact dependence $\HH(\C,\Valpha)$
exhibits similar properties, as demonstrated numerically in
in \fig{valc}. We summarize these properties in the following
\begin{con} \label{con-one}
The function $\HH(\C,\Valpha)$ has an absolute minimum
$\HH^{\hash}=\HH(\C^{\hash},\Valpha^{\hash})$
in $(\C,\Valpha) \in (0,\infty) \times (-\infty,\Vdagger)$,
and for every $\GNaa>\HH^{\hash}$ the level set
$\HH(\C,\Valpha)=\GNaa$ is a simple closed curve, crossing each line
$\C=\const$ (or alternatively each line $\Valpha=\const$) at most twice.
\end{con}

Supposing Conjecture~\ref{con-one} is true, the
solutions of the dispersion relation form a simple closed curve, for
every $\GNaa>\HH^\hash$. Then there exist $\Valpastl$ and
$\Valpasth$ with  $\Valpastl<\Valpasth<\Vdagger$, such that for
any $\Valpha\in(\Valpastl,\Valpasth)$
equation \eq{levelset}
has two positive solutions for $\C$, which we may denote
$\C=\Ch(\Valpha;\Vdagger,\GNaa)$  and
$\C=\Cl(\Valpha;\Vdagger,\GNaa)$, where $0<\Cl<\Ch$.
Since the level sets are simple closed curves the statement may be
inverted so that $\C$ has the role of the
independent parameter: for every $\C$ in some open interval $\C\in(\Cmin,\Cmax)$
there exist two distinct values of the pre-front voltage $\Valpha$ which
satisfy equation \eq{levelset}. At the end points of the interval
\ie~$\C=\Cmin$ and $\C=\Cmax$ equation \eq{levelset} has single
solutions for $\Valpha$. The points $\Cmin$ and $\Cmax$ are extremum
points of $\C$ as a function of $\Valpha$ and the CV restitution curve
has opposite slopes to the left and to the right of each of them.

The fast-time systems is linked, via the boundary conditions
\eq{asympBCs} and parameters $\Valpha$ and $\Vomega$ to the slow-time
system the solution of which is discussed below.

\begin{table}[t]
  \centering
  \newcommand{\Strut}{%
    \vphantom{$\displaystyle\frac{\frac{\frac11}{\frac11}}{\frac{\frac11}{\frac11}}$}%
  }
  \newcommand{\Exp}[1]{e^{#1}}
  \begin{tabular}{|l|l|l|l|l|l|l|}\hline
    $\jp$ & $\zt$ & $\ndel_{\jp}$ & $\Edmp_{\jp}$ & $\Edel_{\jp}$ & $\EinfEdmp_{\jp}$       & $\ninf_{\jp}$
  \Strut\\\hline
    $1$ & $[0,\tast]$        &
  $\displaystyle\frac{1-\Exp{-\Ftwo(\BCL-\tdagger)}}{\Exp{-\Ftwo\BCL}-1}$
                       & $\kthree$  & $\Gtwor$  & $\kthree\Ethree$ & 1
  \Strut\\\hline
    $2$ & $[\tast,\tdagger]$ &
  $\displaystyle\frac{1-\Exp{-\Ftwo(\BCL-\tdagger)}}{\Exp{-\Ftwo\BCL}-1}$
                       & $-\ktwo$   & $\Gtwor$  & $-\ktwo\Etwo$    & 1
  \Strut\\\hline
   $3$ & $[\tdagger,\BCL]$  &
  $\displaystyle\frac{1-\Exp{\Ftwo\tdagger}}{\Exp{-\Ftwo\BCL} -1}$
                       & $\kone$    & $\Gtwol$  & $\kone\Eone$     & 0
  \Strut\\\hline
  \end{tabular}
  \caption{Values of the constants used in the expression \eq{nsol}--\eq{Esol}. }
  \label{tab:const}
\end{table}

\subsection{The slow subsystem}
The slow-time problem \eq{caric-slow} and \eq{asympBCs-slow} can be
solved exactly, too. Notice that it does not depend on the wave
velocity and it is therefore similar to the slow problem considered in
\cite{Biktashev-etal-2008}. First, we solve equation \eq{slow.n},
which is separable and independent of $\E$. Its right-hand side
is different in each of the two intervals $\zt\in [0,\tdagger]$  
$\zt\in[\tdagger,\BCL]$ due to the presence of a Heaviside
function. The equation is constrained by the continuity condition at 
$\zt=\tdagger$ and periodic boundary conditions at $\zt=0$ and
$\zt=\BCL$. The solution has the form
\begin{equation}
  \n(\zt)=\njot(\zt)=\ninf_{\jp}+\ndel_{\jp}\exp(-\Ftwo\zt),  \eqlabel{nsol}
\end{equation}
where the values of the constants $\ninf_{\jp}$ and $\ndel_{\jp}$ are different for
the intervals $\zt\in [0,\tdagger]$ ($\jp=1,2$) and $\zt\in[\tdagger,\BCL]$ ($\jp=3$)
and are given in Table~\ref{tab:const}, and the overset index here
above $\n$ and below above $\E$ designates the corresponding interval.
This solution is then substituted in \eq{slow.E} which becomes
\begin{gather}
  \eqlabel{slow.E1}
  \E'+\Edmp_{\jp}\E
  = \Edel_{\jp} \sum_{\l=0}^{4}\dbinom{4}{\l} \ninf_{\jp}^{(4-l)}\big(\ndel_{\jp}\exp(-\Ftwo\zt)\big)^\l
  + \EinfEdmp_{\jp},
\end{gather}
where the constants $\Edmp_{\jp}$, $\Edel_{\jp}$ and $\EinfEdmp_{\jp}$
also depend on the interval and are given in 
Table~\ref{tab:const}. The values in the table are obtained by a
straightforward manipulation of the model definitions of 
$\Gtilde(\E)$ and $\tildegtwo(\E)$ in \eq{caricfunctions} and the
binomial theorem is used for the term $\n(\zt)^4$.
Within each of the intervals $[0,\tast]$, $[\tast,\tdagger]$ and
$[\tdagger,\BCL]$, equation \eq{slow.E1} is a first order linear ODE
with constant coefficients, and the solution of \eq{slow.E} can be
written in the explicit form,
\begin{gather}
  \eqlabel{Esol}
  \E(\zt)=\Ejot(\zt)=\frac{\EinfEdmp_{\jp}}{\Edmp_{\jp}}
  +\arb_{\jp}\exp(-\Edmp_{\jp}\zt)
  +\Edel_{\jp}\sum_{\l=0}^{4}\dbinom{4}{\l} \, \ninf_{\jp}^{4-l} \, \ndel_{\jp}^\l\,
	\frac{\exp(-\l\Ftwo\zt)}{\Edmp_{\jp}-\l\Ftwo},
\end{gather}
where $\arb_{\jp}$ are integration constants.
The exact solutions \eq{nsol} and \eq{Esol} contain seven
parameters, namely $\tast$, $\tdagger$, $\BCL$, $\arb_1$, $\arb_2$,
$\arb_3$, $\Valpha$, that need to be found from the boundary
conditions and from the internal matching conditions,
\begin{gather}
  \Eeins(0)=\Vomega\,(\ENa-\East)+\East, \quad   \Eeins(\tast)=\East, \nonumber \\
  \eqlabel{alg}
  \Ezwei(\tast)=\East, \quad  \Ezwei(\tdagger)=\Edagger, \\
  \Edrei(\tdagger)=\Edagger, \quad
  \Edrei(\BCL)=\Valpha\,(\ENa-\East)+\East . \nonumber
\end{gather}
The existence of solution to the slow-time
boundary value problem ultimately depends on the existence of
solutions of the transcendental equations \eq{alg}.
Based on the uniqueness and existence of solutions to an initial-value problem,
it is obvious that the problem reduces to four essential unknowns,
which we denote $\Eomega=\Vomega\,(\ENa-\East)+\East=\E(0)$,
$\Ealpha=\Valpha\,(\ENa-\East)+\East=\E(\BCL)$,
$\nstar=\n(0)=\n(\BCL)$ and $\BCL$.

\sglfigure{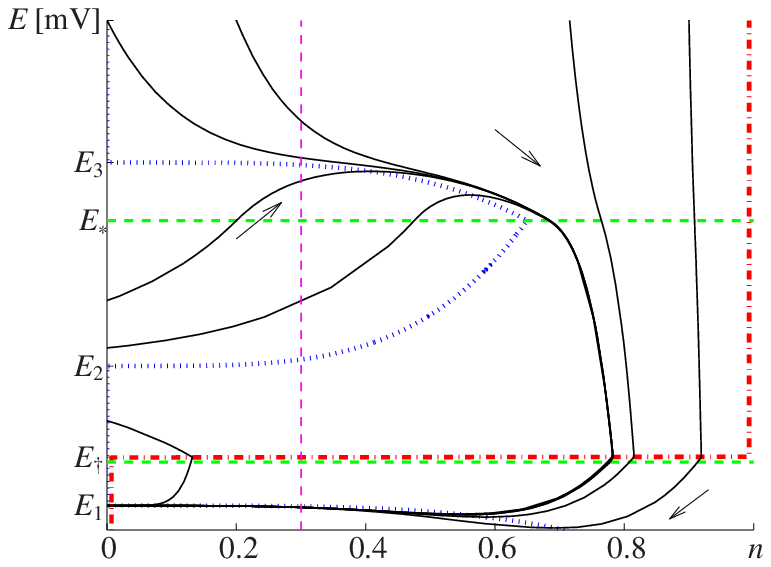}{
 (color online)
  The phase portrait of the slow subsystem \eq{caric-slow} of the
  Caricature Noble model.  Standard parameter values
  \eq{caricfunctions} are used except $\Ftwo$ which increased three
  times to $\Ftwo=1/90\, \ms^{-1}$, for
  visualization purposes. Red dash-dotted lines represent vertical
  isoclines $\d\n/\d\t = 0$. Blue dotted lines   represent horizontal
  isoclines $\d\E/\d\t = 0$. Thin solid black   lines with attached
  arrows represent trajectories.   
}{ne}

\begin{prop}
\label{prop-two}
Suppose that the values of the parameters of the slow subsystem
\eq{caric-slow} and \eq{asympBCs-slow} obey the same qualitative
relationships as the default values, \ie\
$\Eone<\Edagger<\Etwo<\East<\Ethree$,
$\kone,\ktwo,\kthree,\Ftwo>0$,
$\Gtwol,\Gtwor<0$ and $\Gtilde(\E)$ is continuous.
Then for any $\nstar\in(0,1)$ and $\Eomega\ge\Edagger$, the system of
equations 
\eq{nsol}, \eq{Esol}, and \eq{alg}
has a unique solution for  $\Ealpha$
and $\BCL$.
For a fixed $\nstar$, this defines a function $\Vomega\mapsto\Valpha$
with domain $\Vomega \in (\Vdagger,+\infty)$ and range $\Valpha \in(-\infty,\Vdagger)$,
which is monotonically decreasing.
\end{prop}
\begin{proof}
is evident from the phase portrait shown in \fig{ne}.
Every trajectory starting above $\E=\Edagger$ goes to the right
and therefore eventually goes down below $\Edagger$. Whilst below
$\Edagger$ it goes to the left so $\n(\zt)\to0$ as $\zt\to\infty$ (this is
also evident from the analytical solution).  Therefore there
exists a point $\zt$ such that $\n(\zt)=\n(0)=\nstar$.
Moreover, such $\zt$ is unique, as the domain $\E<\Edagger$ is absorbing
and $\n(\zt)$ is monotonically increasing outside it and
monotonically decreasing inside it.  
So we have the proposed mapping
$(\nstar,\Vomega)\mapsto(\Ealpha,\BCL)$.
The monotonicity of this mapping follows from the fact that the
trajectories cannot intersect, 
so if $\Eomega_2>\Eomega_1$, then the contour made by the straight
line between points $(\nstar,\Ealpha_1)$ and $(\nstar,\Eomega_1)$ and the segment of
trajectory joining these two points, lies within the similar contour made by
points $(\nstar,\Ealpha_1)$ and $(\nstar,\Eomega_1)$.
\end{proof}

\subsection{Matching and well-posedness}
The CV restitution curve can be obtained by combining the results
of the slow and the fast subsystems. According to
Proposition~\ref{prop-one} and Conjecture~\ref{con-one}, for
sufficiently large values of the excitability $\GNaa$ the fast 
subsystem defines the wave velocity as a function of prefront voltage,
$\C=\C^{\pm}(\Valpha)$, for $\Valpha\in(-\infty,\Vdagger)$, which, via
equation \eq{Vomega}, defines the post-front voltage
$\Vomega=\Vomega^{\pm}(\Valpha)$.
On the other hand, according to Proposition~\ref{prop-two}, the wave
period $\BCL$ and the peak voltage $\Eomega=\East+\Vomega(\ENa-\East)$
are functions of $\Ealpha=\East+\Valpha(\ENa-\East)$ and $\nstar$.
Hence, the matching of the fast and slow solutions, for any given
$\nstar$, can be obtained by solving the simultaneous system of
equations for $\Valpha$ and $\Vomega$, one resulting from the
fast subsystem and one resulting from the slow subsystem, the
latter depending on $\nstar$ as a parameter.  This subsequently provides
where $\C(\Valpha)$ is given by the fast subsystem, and %
$\BCL=\BCL(\nstar,\Vomega)$ from the slow subsystem.
Hence, we have ultimately the CV restitution curve
$(\c(\nstar),\BCL(\nstar))$ in parametric form and we have proven
that the asymptotic CV-restitution problem
\eq{caric-slow},\eq{caric-fast} and \eq{asympBCs} 
is indeed well-posed.

\dblfigure{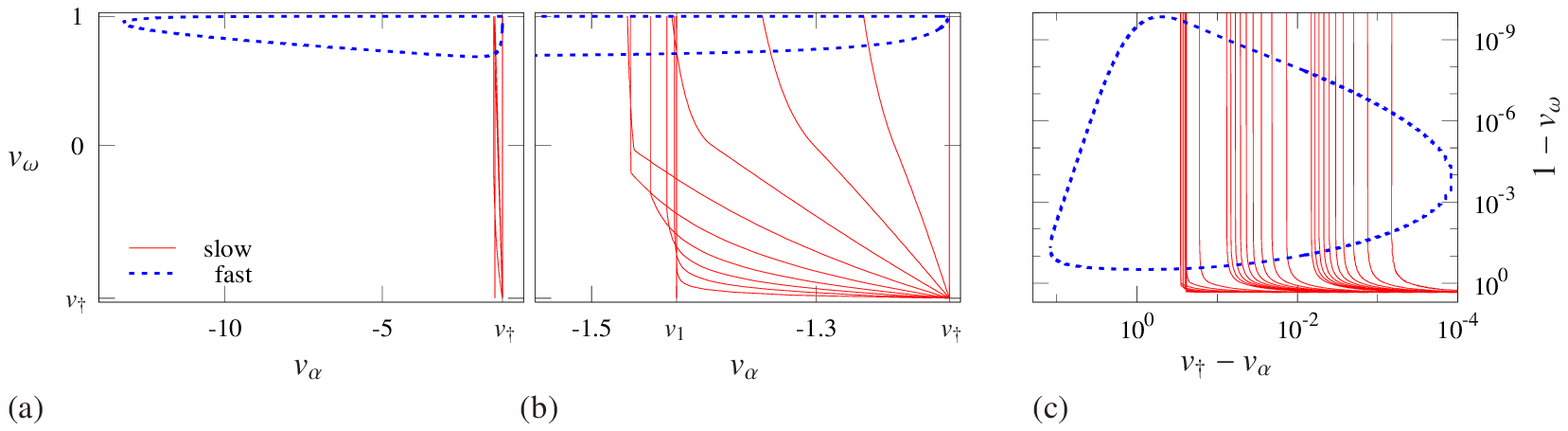}{
  (color online)
  Matching solutions of the fast subsystem eq{caric-fast} and slow 
  subsystem \eq{caric-slow} of the Caricature Noble model in terms of the
  dimensionless pre- and post-front voltages $\Valpha$ and $\Vomega$.
  Standard parameter values \eq{caricfunctions} except
  $\Ftwo=1/90\, \ms^{-1}$ as in \fig{ne}.
  The non-dimensionalized resting potential
  $\Vone=(\Eone-\East)/(\ENa-\East)\approx-1.4242$
  and the non-dimensionalized $\h$- and $\n$-gate switch potential
  $\Vdagger=(\Edagger-\East)/(\ENa-\East)\approx-1.1818$.
  The choice of values of $\nstar$ in (a) is: $0$, $0.5$ and $1.0$; for
  in (b): from $0$ to $1$ with step $0.1$;
  in (c): from $0$ to $0.9$ with step $0.1$,
  then to $0.99$ with step $0.01$,
  and then to $0.999$ with step $0.001$.
}{valvom}

The equations of the aforementioned system for $\Valpha$ and $\Vomega$ are
complicated and finding analytical solutions seems seems unfeasible. Some
qualitative insight can be obtained from numerical analysis, which for
the standard parameter values is illustrated in \fig{valvom}.  The
solutions correspond to the intersection of the closed fast-subsystem
contour (dashed blue line) with a slow subsystem line (solid red
lines). The family of slow-subsystem lines stretches continuously, but
non-monotonically, from the vertical line $\Valpha=\Vone$ at
$\nstar=0$ to the vertical line $\Valpha=\Vdagger$ for $\nstar=1$. In
accordance with Proposition~\ref{prop-two}, these lines are
monotonically decreasing except for the above mentioned vertical lines
for extreme values of $\nstar$. Almost all of these lines intersect
the fast-subsystem contour, with the exception of the lines with
$\nstar$ very close to unity. The line $\Valpha=\Vone$, $\nstar=0$
corresponds to the limit $\BCL\to\infty$ of the restitution curve, as
that limit corresponds to a solitary wave propagating through the
resting state, which is characterized by $\E=\Eone$ and $\n=0$. The
other extremity corresponds to the $\nstar$ line which is tangent to
the slow-subsystem contour for a value of $\nstar$ very close to but
smaller than unity, and at a value of the pre-front voltage 
$\Valpha$ very close to but smaller than that of the parameter $\Vdagger$.
As evident from the above analysis, \eg\ see
\fig{valc}(a), we have $\Valpha^*\to\Vdagger-0$ as $\GNaa\to\infty$,
which motivates consideration of asymptotics related to this
limit. The details are presented in the next subsection.

Another qualitative feature evident from \fig{valvom} is that the slow
lines are nearly vertical for larger $\Vomega$. This, and the
non-monotonic behaviour of the horizontal ($\Valpha$) position of the
slow-subsystem lines on $\nstar$ around larger values of $\Vomega$ is
a direct consequence of the non-monotonic behaviour of the trajectories
which can be observed in \fig{ne}: a typical ``tail'' of a trajectory
starting at a large $\Vomega$ is a curve which starts at $\Vone$ at
$\n=0$, then decreases and then increases up until $\Vdagger$;
besides, all trajectories starting from larger $\Vomega$ join together
very closely.
As follows from the analysis in \cite{Biktashev-etal-2008}, this is
due to an extra small parameter present in the slow subsystem, namely
$\Ftwo\to0$.

These observations motivate consideration of further asymptotics to the
obtained solution, which lead to less accurate but more explicit
results. We present them less formally than the main limit $\eps\to0$
as they are of a secondary importance to our main results.

\subsection{Further asymptotics}

\paragraph{The fast branch of the restitution curve} For
values of $1-\n\sim1$, the period $\BCL$ is large compared to
the parameter $\Ftwo^{-1}$ which plays the role of a time
constant in the $\n$-gate equation \eq{caric.n}, and hence we can
exploit the Tikhonov singular  perturbation in terms of $\Ftwo$
understood as a small parameter. This corresponds to the
secondary asymptotic embedding $\epstwo\to0$ considered in
\cite{Biktashev-etal-2008}. In this limit, the trajectories differ
only in the initial post-overshot stage, after which they move along
the reduced superslow manifold 
$\n\approx\Nss(\E)\equiv\left(-\Gtilde(\E)/\tildegtwo(\E)\right)^{1/4}$
with the exception of the repolarization stage when
$\Nss(\E)\notin\Real$. It is important to note that except during
the the initial transient, the trajectories are nearly the same, up
to a correction which is exponentially small in $\Ftwo$ as evident
from \fig{ne}. We consider two parts of a typical trajectory, one for
$\zt\in[0,\tdagger]$ when $\n$ increases, and the other for
$\zt\in[\tdagger,\BCL]$ when it decreases. 
Due to the above mentioned convergence of trajectories, the
value $\nmax=\n(\tdagger)$ is practically independent on $\Eomega$
up to exponentially small corrections. A typical value of $\nmax$ can
be found \eg\ in the following way: first, consider solution
\eq{Esol} for $\jp=1$ and  $\arb_1=0$ and solve the matching
condition $\Eeins(\tast)=\East$ for $\tast$, next determine
$\arb_2$ from the initial condition $\Ezwei(\tast)=\East$, then solve
$\Ezwei(\tdagger)=\Edagger$ for $\tdagger$, and finally with the 
knowledge of $\tdagger$ and $\nstar$, the value of $\nmax$ can be
obtained from \eq{nsol} as
$\nmax=\nzwei(\tdagger)=\ndrei(\tdagger)=1-(1-\nstar)\,e^{-\Ftwo\tdagger}$. 
This however leads to a transcendental equation. Numerical value
for the standard parameter values is $\nmax\approx0.7827$.

Using  \eq{nsol}, the duration of the second half of the
trajectory, between $\zt=\tdagger$, $\n=\nmax$, $\E=\Edagger$ and
$\zt=\BCL$, $\n=\nstar$, $\E=\Ealpha$, is given by
\begin{equation}
  \ttwo = \BCL-\tdagger = \frac{1}{\Ftwo} \,\ln\left(
    \nmax + (1-\nmax)\,e^{\Ftwo\BCL}
  \right)                                \eqlabel{t2-bcl}
\end{equation}
The evolution of $\E$ during the second half of a typical slow
trajectory, is described by the relevant form of equation
\eq{slow.E}
\[
 \Df{\E}{\zt} = \Gtwol\nmax^4\,e^{-4\Ftwo(\zt-\tdagger)} +
 \kone(\Eone-\E)
\]
wherefrom
\begin{equation}
  \Ealpha = \Eone
  + \frac{\Gtwol\nmax^4}{\kone-4\Ftwo}\,e^{-4\Ftwo\ttwo}
  +
  \left(\Edagger-\Eone-\frac{\Gtwol\nmax^4}{\kone-4\Ftwo}\right)e^{-\kone\ttwo} .
  \eqlabel{Ealpha-t2}
\end{equation}
Combining \eq{t2-bcl}, \eq{Ealpha-t2} and \eq{Climfast}, we get an
explicit dependence of $\c(\BCL)$ in elementary functions.

\paragraph{The slow branch of the restitution curve} The slow
branch is considered in a similar way. The 
difference is in the initial transient where a typical trajectory
approaches the reduced superslow manifold from lower values of $\E$ rather than from
the higher $\E$ as it was for the faster branch, and also in the
$\c(\Valpha)$ dependence. Hence the dependence of $\c(\BCL)$ is
obtained by combining \eq{t2-bcl}, \eq{Ealpha-t2} and \eq{Climslowas}.

\paragraph{The turning point of the restitution curve} The
turning point is the point where the fast branch meets the slow
branch. It is characterized by extreme proximity of $\Valpha$ to
$\Vdagger$ and $\nstar$ to 1. The front parameters can be estimated
via the limit $\Valpha\to\Vdagger$, $\GNaa\to\infty$ of
\eq{C-ast} and \eq{HHlwr-ast}, which gives the highest pre-front
voltage as 
\begin{equation}
  \Valpha^{*} \approx \Vdagger \left( 1 + \left(-\Vdagger(1-\Vdagger)/\GNaa\right)^e \right)
                                        \eqlabel{Valpha-ast}
\end{equation}
and the corresponding
slowest stable front velocity as
\begin{equation}
  \C^*=\C^{*}(\Valpha^*) \approx \left(e\ln\left(\frac{\GNaa}{-\Vdagger(1-\Vdagger)}\right)\right)^{1/2}.
  \eqlabel{C-ast-ast}
\end{equation}
Given $\Valpha^*$ and $\C^*$, equation \eq{Vomega} then gives the
value of the post-front voltage,
\begin{gather}
\Vomega^*=
  1-\frac{\Big(\so^*/2\Big)^{(\C^*)^2}}{\Gamma\left((\C^*)^2+1\right)\,\BesselI_{(\C^*)^2}\Big(\so^*\Big)},
                                        \eqlabel{Vomega-ast}
\\
  \so^*=2\C^*\sqrt{\GNaa}\left(1-\Vdagger/\Valpha^*\right)^{1/(2(\C^*)^2)} .
                                        \eqlabel{s0-ast}
\end{gather}
The duration of the slow
trajectory corresponding to these $\Valpha^*$ and $\Vomega^*$ will be
an estimate of the shortest wave period possible in this model,
$\BCL^*$.  A simple approximation of it can be obtained from the
consideration that in the limit $\Valpha\approx\Vdagger$, we have
$\n(\zt)\approx1$ throughout, hence the slow-subsystem equation for
$\E$ is simplified by replacing $\n(\zt)=1$ so the link between $\n$
and $\E$ dynamics is only via the values of $\tdagger$ and $\BCL$.
The period $\BCL$ is almost the same as the $\tdagger$ taken by the
voltage to decrease from $\Eomega$ to $\Edagger$, since the interval
$\BCL-\tdagger$ needed to decrease from $\Edagger$ to $\Ealpha$ is
small compared to that. In this case, the interval $\tdagger$ can be
estimated from solutions \eq{Esol} for $\jp=1,2$. Hence, we have
\begin{gather}
  \BCL^*\approx \frac{1}{\ktwo} \ln\left(\frac{
      -\Gtwor + \ktwo(\Etwo-\Edagger)
    }{
      -\Gtwor +\ktwo(\Etwo-\East)
    }\right) 
  + \frac{1}{\kthree} \ln\left(\frac{
      \kthree(\Eomega^*-\Ethree)-\Gtwor
    }{
      \kthree(\East-\Ethree)-\Gtwor
    }\right).                           \eqlabel{BCL-ast}
\end{gather}
Equations \eq{Valpha-ast}--\eq{BCL-ast},
together with the scaling relationships \eq{nondim}
define the turning point of the restitution curve.  For the standard
parameter values, this gives

\begin{gather}
\C^*\approx 2.973, 
\quad
\Valpha^*\approx-1.18199, 
\quad 
\Valpha^*-\Vdagger\approx1.709\cdot10^{-4}, 
\nonumber\\
\c^*\approx 2.102\, \su/\ms, 
\quad
\BCL^*\approx13.09\, \ms. 
                                    \eqlabel{tip}
\end{gather}

\dblfigure{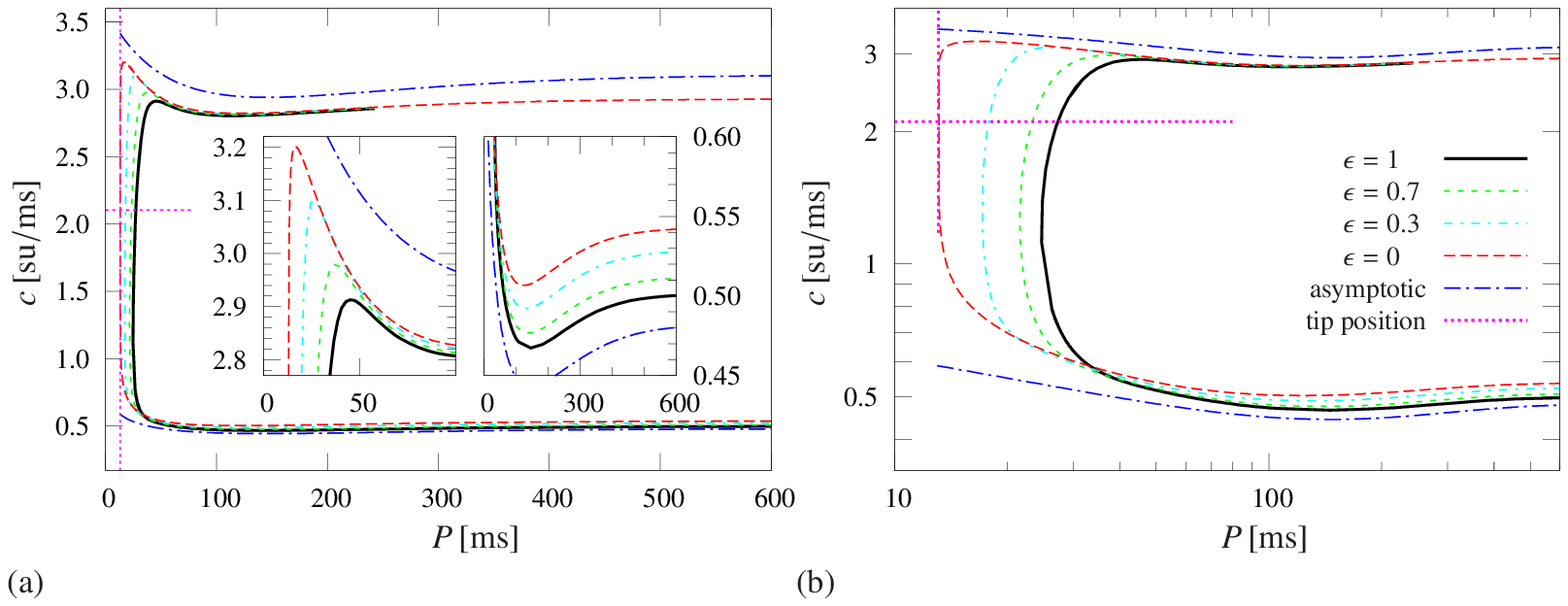}{
  (color online)
  Restitution curves of the Caricature Noble model,
  (a) in Cartesian and (b) logarithmic coordinates.
  Insets in panel (a) show selected features magnified. 
  Lines in all plots as described by the legend in panel (b),
  specifically, 
  accurate numerical solution of the full boundary-value
  problem \eq{full} for various values of
  $\eps>0$ (bold solid black, short-dashed green and dash-dotted cyan for $\epsilon=1,0.7, 0.3$ respectively),
  the accurate matching of the fast and slow subsystems
  \eq{caric-slow}--\eq{asympBCs-slow} which corresponds to the limit
  of $\eps=0$ (long-dashed red),
  and the approximations of the fast and slow branches of the RC at
  $\eps=0$ given by asymptotics \eq{t2-bcl}, \eq{Ealpha-t2},
  \eq{Climfast} and \eq{Climslowas} (dash-dotted blue).
  The magenta dotted lines indicate the position of the tip of the
  curve $\Valpha^*$ as estimated by \eq{tip}. 
}{pc}

\sglfigure{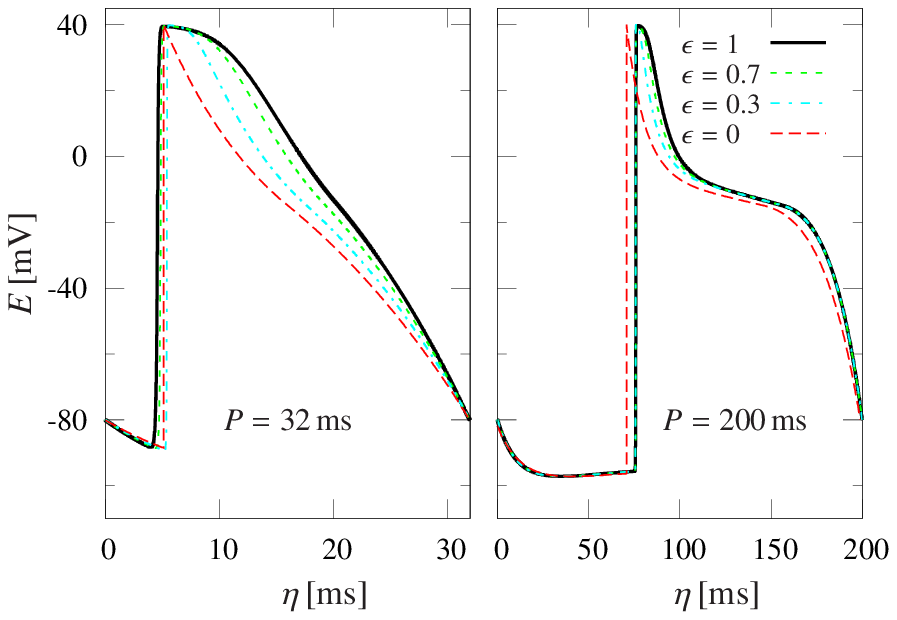}{
  (color online)
  Profiles of the action potentials corresponding to the faster branch
  of the solution of the Caricature Noble model. 
  The values of the period $\BCL$ and the values of $\eps$
  are given in the plot. The values $\eps\ne0$ correspond to the
  full CV restitution boundary value problem  \eq{full} for \eq{caric}
  with standard parameter values while $\eps=0$ corresponds to the asymptotic
  limit given by \eq{caric-slow}, \eq{caric-fast} \eq{asympBCs}.
}{ze}

\subsection{Comparison of the asymptotics with the exact solution}

With the aim to assess the proximity of the analytical solutions
obtained in the main asymptotic limit $\eps\to0$ and  the full
numerical solutions of the CV restitution boundary value problem, we
present in \fig{pc} sets of CV restitution curves of the Caricature
Noble model, and we show in \fig{ze} the action
potential profiles for two selected base cycle lengths $\BCL$.
We also demonstrate in \fig{pc} the asymptotic estimates of
the upper and lower branches and of the tip position $\BCL^*$ of the
curve found with the help of the secondary embedding $\Ftwo\to0$.
The asymptotic CV restitution curve was obtained by solving numerically
problem \eq{caric-slow}, \eq{caric-fast} and \eq{asympBCs} which
defines $\c(\Valpha)$ and $\BCL(\Valpha)$.  The full CV restitution
curve was obtained by solving the full boundary-value problem
\eq{full} formulated for equations \eq{caric}. \Fig{pc}(a) presents the
curves in Cartesian coordinates and \fig{pc}(b) in semi-logarithmic
coordinates to reveal in more details the behaviour at small
values of the wave period $\BCL$. We can see that as $\eps$ is decreased
the solution of the full problem converges to the solution
of the asymptotically reduced problem at $\eps=0$, and that the 
full model curve for $\eps=1$ at standard parameter values is
close to the asymptotic limit everywhere except at the smallest
values of the period $\BCL$. 

\sglfigure{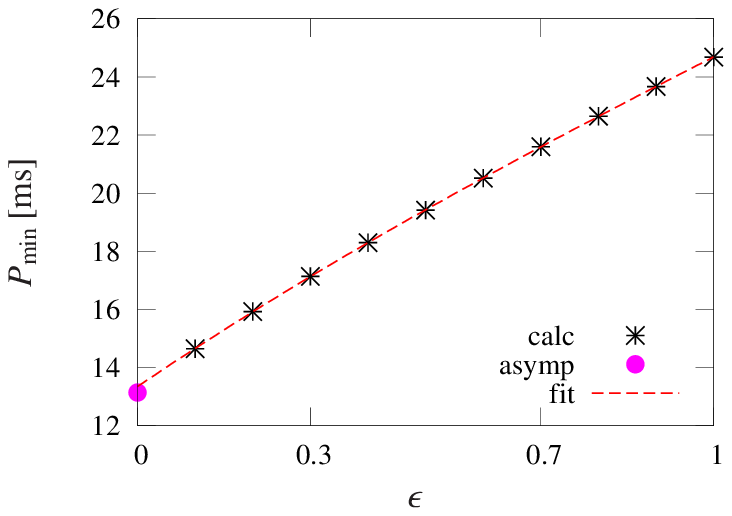}{
  (color online)
  Dependence of the minimal basic cycle length $\BCLmin$ on the embedding parameter $\eps$. 
  The asterisks represent $\BCLmin$ calculated at the specified values of $\eps$, the solid circle is the
  estimate \eq{tip}. The dashed line is the best cubical fit of all the asterisk points, which is
  $\BCLmin=13.34
  +13.40
  \eps-2.91
  \eps^2+0.85
  \eps^3\,[\ms]$.
}{pmin}

This indicates that at small $\BCL$, the parameter $\eps$ is not a
``good small parameter''. Note that in our asymptotic analysis, we
have calculated the period $\BCL$ as the length of the slow subsystem
solution, and we have neglected the contribution of the fast
subsystem, \ie\ the duration of the front, which is small of the order
$\O{\eps}$. However, at the smallest values of $\BCL$, the the front
length is comparable to duration of the solution of the slow
subsystem. This can be seen already in \fig{pc} and \fig{ze}, and is
further confirmed by the analysis of the dependence of the minimal
wave period $\BCL$ on $\eps$ shown in \fig{pmin}. We see that to
second order, the basic cycle length can be approximated by
$\BCL=\BCL_0+\eps\BCL_1+\O{\eps^2}$, where $\BCL_0$ is the cycle
length given by the asymptotic theory presented above, and
$\BCL_1\approx13.4$ may be interpreted as the front duration in the
fast subsystem. Note that $\BCL_0\sim\BCL_1$, hence neglecting
$\O{\eps}$ for smaller $\BCL$ produces relatively large error.  That
this is not the whole story, however, as not only the horizontal
position of the $\BCL^*=\BCLmin$ point changes with $\eps$, but also
its vertical position $\c^*$, so a proper next-order asymptotic should
take into account of the influence of the slow subsystem on the front
velocity as well.

\section{Asymptotic restitution curves in the Beeler-Reuter model}
\label{applicationBR}

\dblfigure{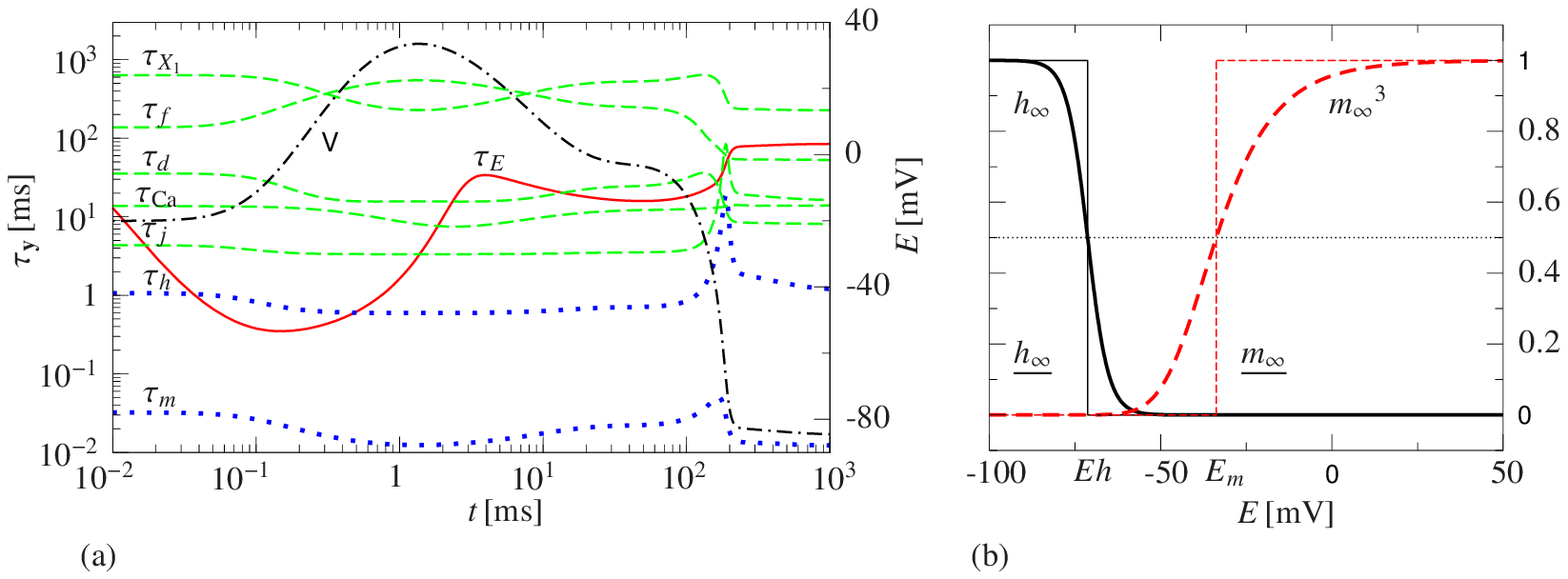}{
  (color online)
  (a) On the left ordinate: Time scaling functions $\tauy$ of the dynamical variables of the
  Beeler-Reuter model during a typical action potential. The time
  scaling function of the voltage $\tauE$ is given by a solid line
  (red online), the time scaling functions of the slow variables are
  given by dashed lines (green online) and of fast variables by dotted
  lines (blue online) with all lines labeled correspondingly in the
  plot. On the right ordinate: Voltage of a typical action potential
  given by a dash-dotted line. (b) The quasistationary values $\mbarcub$
  and $\hbar$ and their approximations $\emb{\mbar}(E,0)$ and
  $\emb{\hbar}(E,0)$.
}{bremb}

In this section we apply the methodology presented above to the
Beeler-Reuter ventricular model \cite{Beeler-Reuter-1977}.  This model is an example
of a realistic voltage-gated model which represents an intermediate
step between relatively simple early models and complicated
contemporary models. It has played an important role for
understanding of cardiac excitability with a large volume of
literature devoted to it, and it is still the model of
choice in many situations, 
\eg\ \cite{Dokos-Lovell-2004,Herman-etal-2007,Chen-etal-2007} are some recent
examples. In the following, we find the CV restitution curve of the 
Beeler-Reuter model using both the asymptotic formulation
and the full formulation of the periodic boundary value problem
as described above and demonstrate an excellent quantitative
agreement.

\subsection{The model and the asymptotic embedding}

The initial step of our approach requires an appropriate asymptotic
embedding of the original Beeler-Reuter model \cite{Beeler-Reuter-1977}. The embedding
is constructed following the procedures presented in detail in
our earlier works
\cite{Biktasheva-etal-2006,Simitev-Biktashev-2006,Biktashev-etal-2008} 
and here we shall summarize briefly the relevant arguments. We
would like to remark that an analogous embedding procedure applies to
the Caricature Noble model of section \ref{asymptotic} where $\eps$
appeared seemingly without much justification.

We rewrite the Beeler-Reuter model in the one-parameter form,
\begin{subequations}
    \eqlabel{BR}
  \begin{gather}
\df{\E}{\t} = \frac{1}{\eps} \gNa \, (\ENa-\E) \, \emb{\mbar}(\E;\eps)\,\j \,\h +
\Islow(\E,\y) + \eps \ddf{\E}{\x} ,     \eqlabel{BR.E}\\
\df{\h}{\t} = \frac{1}{\eps} \big(\emb{\hbar}(\E;\eps) - h\big)/\tauh(\E),     \eqlabel{BR.h}\\
\df{\j}{\t} = \big(\jbar(\E) - j\big)/\tauj(\E) \eqlabel{BR.j}\\
\df{\y}{\t} = \fy(\E,\y,\dots),    \eqlabel{BR.y}
  \end{gather}
\end{subequations}
where only the equations affected by the artificial small parameter
$\eps\ll 1$ are shown. 
  As in the previous model, the voltage $\E$ is measured in $\mV$,
  time in $\ms$, the gating variables $\h$, $\j$, $\m$, $\y$ are
  non-dimensional, and the space coordinate $\x$ is measured in
  $\su=\ms^{1/2}$.

The functions $\tauy(\E)$ and $\ybar(\E)$ are
time-scaling functions and quasi-stationary values of the gating
variables, respectively. 
Functions $\emb{\mbar}(\E;\eps)$ and $\emb{\hbar}(\E;\eps)$
are ``embedded'', \ie\ they are $\eps$-dependent versions of
$\mbar(\E)$ and $\hbar(\E)$ such that $\emb{\mbar}(\E;1)=\mbar(\E)$
and $\emb{\hbar}(\E;1)=\hbar(\E)$ on one hand and
$\emb{\mbar}(\E;0)= \Heav(\E-\Em)$ and
$\emb{\hbar}(\E;0)=\Heav(\Eh-\E)$ on the other hand,
with $\Em=-33.75\, \mV$ and
$\Eh=-71.33\, \mV$
so that $\mbarcub(\Em)=1/2$ and $\hbar(\Eh)=1/2$.
The last two
parameters are analogous to $\East$ and $\Edagger$ of the Caricature
Noble model. The rest of the model \eq{BR} is the same as defined in  
\cite{Beeler-Reuter-1977}, namely 
$\Islow=-(\INac+\IKi+\Ixi+\Is)$ is the sum of all
slow currents,
$\y=(\xone,\dgate,\fgate,\Ca)$ is the vector of all slow variables in
addition to gate $\j$ which is also slow, and $\fy(\cdot)$ stands for
the functions governing the dynamics of the gating variables $\y$. The
rationale for this parameterisation is the following. 
\begin{enumerate}
\item The dynamic variable $\m$ is a `superfast' variable  and has been
adiabatically eliminated by replacing it with its quasi-stationary
value $\mbar$.
The variables $\E$ and $\h$ are  `fast variables', \ie\
they change significantly during the upstroke of a typical AP
potential, unlike all other variables which change only slightly
during that period. The relative speed of the dynamic variables is
estimated by comparing the magnitude of their corresponding
`time-scaling functions' as illustrated in \fig{bremb}(a). For a
system of differential equations
$\d{\w_\l}/\d{\t}  = \F_\l(\w_1,\dots\w_\N)$, $\l=1,\dots\N$
the time
scale functions are defined as
$\taul(\w_1,\ldots) \equiv \left|\d{\F_{\l}}/\d{\w_{\l}}\right|^{-1}$,
$\l= 1\ldots\N$ and coincide with the
functions $\tauy(\cdot)$ already present in \eq{BR}.
\item The dynamic variable $\E$ is fast due only to
  one of the terms in the right-hand side, the large sodium current
  $\gNa \, (\ENa-\E) \, \emb{\mbar}(\E;\eps)\,\j \,\h$,
  whereas other currents are not that large and so do not have the
  large coefficient $\eps^{-1}$ in front of them.
\item The fast sodium current $\gNa \, (\ENa-\E) \,
  \emb{\mbar}(\E;\eps)\,\j \,\h$ is large only during the upstroke
of the AP, and not that large otherwise. This is due to the fact that
either gate $\m$ or gate $\h$ or both are almost closed outside the upstroke
since their quasistationary values $\mbar(\E)$ and $\hbar(\E)$ are
small there as seen in \fig{bremb}(b). Thus in the limit $\eps\to0$,
functions $\mbar(\E)$ and $\hbar(\E)$ have to be considered zero in
certain overlapping intervals $\E\in(-\infty,\Em]$ and
$\E\in[\Eh,+\infty)$, and $\Eh\le\Em$, hence the representations
$\emb{\mbar}(\E;0)=\Heav(\E-\Em)$ and $\emb{\hbar}(\E;0)=\Heav(\Eh-\E)$.
\end{enumerate}
A more detailed discussion of the parameterisation given by
\eq{BR}, as well as the justification of our
method of parametric embedding, \ie\ a seemingly ``arbitrary''
introduction of an artificial small parameter $\eps$, can be found in 
\cite{Biktashev-etal-2008}.

\subsection{The asymptotic reduction}

We are now ready to formulate the asymptotic CV restitution problem as
a set of coupled boundary value problems similar to those described in
section \ref{AsymptFormulation}.  Being
interested in propagation with constant velocity and fixed shape, we
introduce the travelling wave coordinates $\z=\x-\c\,\t$ for the slow 
subsystem and $\Z=\X-\c\,\T=\z/\eps$ for the fast subsystem.  As 
before, we distinguish the functions of the fast time by name and
set $\E(\z)=\Efast(\Z)$ and $\h(\z)=\hfast(\Z)$.  

The slow-time subsystem is obtained in the limit $\eps\to0$ of the
original slow independent variable $\z$,
\begin{subequations}
    \eqlabel{BRbvpslow}
  \begin{gather}
    \c\Df{\E}{\z} +    \Islow(\E,\y) = 0, \\
    \c\Df{\j}{\z} +     \frac{\jbar(\E) - \j}{\tauj(\E)} = 0, \\
    \c\Df{\y}{\z} +     \fy(\E,\y,\dots) = 0, \\
\E(\z=0) = \Ealpha, \quad \E(\z=\c\BCL)=\Eomega, \nonumber \\
\j(\z=0)=\j(\z=\c\BCL)=\jalp, \\
\y(\z=0)=\y(\z=\c\BCL). \nonumber 
  \end{gather}
\end{subequations}
The fast-time subsystem is obtained in the limit $\eps\to0$ of the
fast independent variable $\Z$,
\begin{subequations}
    \eqlabel{BRbvpfast}
  \begin{gather}
\DDf{\Efast}{\Z} + \c\Df{\Efast}{\Z} + \gNa(\ENa-\Efast)\,\jfast\,\Heav(\Efast-\Em)\,\hfast = 0,  \\
\c\Df{\hfast}{\Z} + \frac{\Heav(\Efast-\Eh) - \hfast}{\tauh(\Efast)} = 0, \\
\Efast(-\infty) =  \Efastalpha, \ \
\Efast'(+\infty) =    0, \ \
\Efast(0)= \Eh, \nonumber\\
\Efast(+\infty)=\Efastomega, \ \
\hfast(-\infty) = 1, \ \
  \end{gather}
\end{subequations}
The boundary conditions of the fast subsystem include the pinning
condition eliminating the translational invariance along the $\Z$
axis.  This problem depends on four parameters, namely the pre-front
voltage $\Efastalpha$, the post-front voltage $\Efastomega$,
the fixed value of the $\j$ gate inherited from the slow system
$\jfast$ and the wave speed $\c$ which are determined by
matching with the slow subsystem, given by the conditions,
  \begin{gather}
  \Ealpha=\Efastalpha, \quad
  \Eomega=\Efastomega, \quad
  \jalp=\jfast.
    \eqlabel{BRbvpmatch}
  \end{gather}

\subsection{The fast subsystem}

\dblfigure{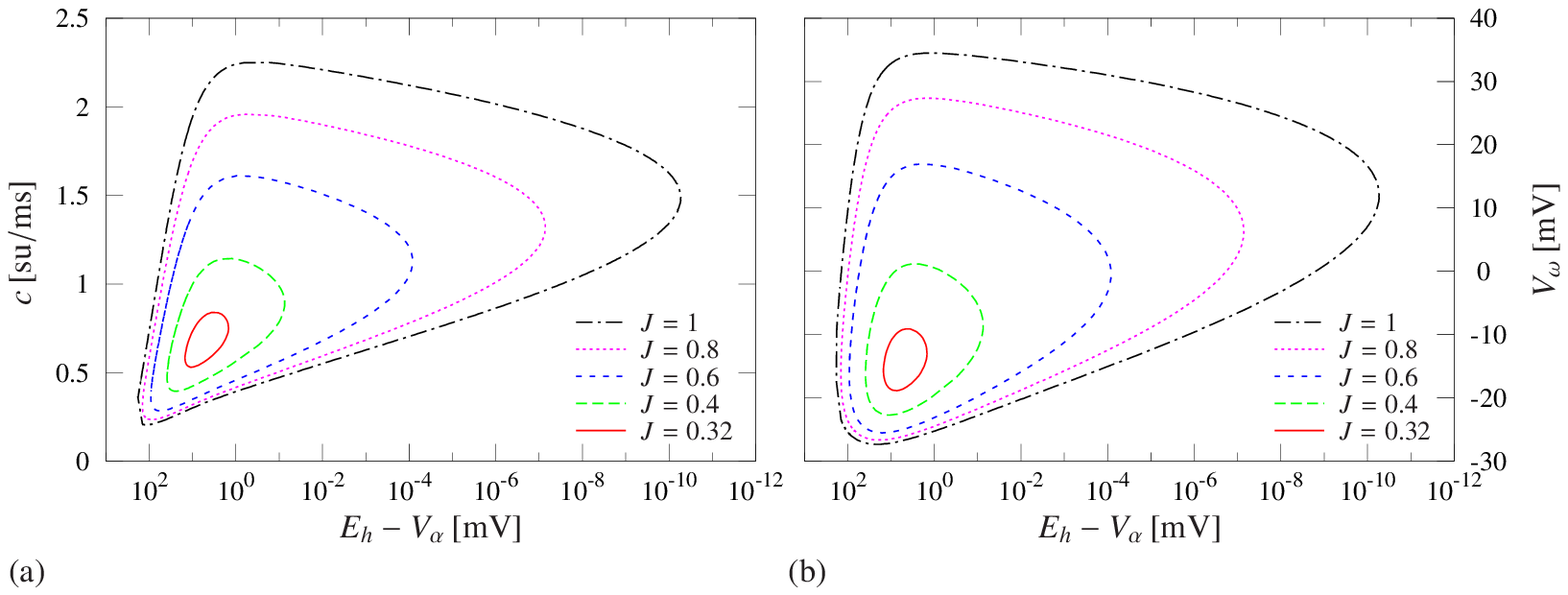}{
  (color online)
  Solutions of the fast subsystem \eq{BRbvpfast} of the
  Beeler-Reuter mode for a  selection of values of $\jfast$. 
  (a) Front speed vs prefront voltage.
  (b) Postfront voltage vs prefront voltage. 
  In both panels, the prefront voltage is shown in logarithmic scale
  with respect to $\Eh$, since the curves are very close to
  the line $\Efastalpha=\Eh$. Note that $\Efastalpha$ increases from left to
  right. 
}{braoc}

The fast subsystem \eq{BRbvpfast} describes the wave front of an
action potential as it propagates by diffusion. This problem
has, on one hand, differential equations of cumulative order three and
four parameters, and on the other hand, five constraints, thus the
solution will typically depend on two ``leading'' parameters chosen
arbitrarily, and the other parameters will be functions of these
two. The structure of the equations \eq{BRbvpfast} is very similar to
the fast subsystem of the Caricature Noble model \eq{caric-fast} and
\eq{asympBCs-fast}. In fact, if $\tauh(\Efast)$ was a piecewise
constant function with a step at $\Efast=\Eh$, then the Beeler-Reuter
fast subsystem would be equivalent to the Caricature Noble fast
subsystem up to parameter values and identification of $\gNa\jfast$
in the former with $\gNa$ in the latter. Hence we may expect that the
set of solutions here has a structure similar to that of the
Caricature Noble model. In particular, we 
expect that $\jfast$ and $\Efastomega$ can be determined as univalued
functions of $\Efastalpha$ and $\c$. Further, we expect that for a fixed
$\jfast$, we have the set of solutions in the $(\Efastalpha,\c)$ such
that there exists $\jmin$ such that if $\jfast>\jmin$ then there
exists an interval $(\Efastalphamin(\jfast),\Efastalphamax(\jfast))$
such that for any $\Efastalpha$ within this interval, there are two
solutions for the front velocity $\c=\c^{\pm}(\jfast,\Efastalpha)$, and
correspondingly two values of the postfront voltage
$\Efastomega=\Efastomega^{\pm}(\jfast,\Efastalpha)$. This is confirmed
by numerical analysis of this problem, see \fig{braoc}.

\subsection{The slow subsystem}

As in the Caricature Noble model, the slow subsystem in the leading
order does not depend on diffusion, and therefore coincides,
up to the scaling of the independent variable,
with the slow subsystem of the single-cell model.
The slow subsystem depends on four parameters, namely the
pre-front voltage  $\Ealpha$, the post-front voltage
$\Eomega$, the initial value of the $\j$-gate $\jalp$ and the
wave   period $\BCL$, and has differential equations of cumulative order 
of $\dim(\y)+2$. On the other hand it has $\dim(\y)+4$
constraints, hence, similarly to the fast system, it has typically
a two-parametric family of solutions. From the viewpoint of matching
with the fast-time problem and in analogy with the Caricature Noble
case, one possible convenient choice of leading parameters is
$\Ealpha$ and $\jalp$ from which $\Eomega$, $\BCL$ as well as $\y(0)$
can be found. However, unlike the Caricature Noble case, it is now
more difficult to establish rigorous conditions for existence of
solutions. This, however, can be easily done numerically.

\subsection{Matching and comparison with the exact solution}

The three constraints \eq{BRbvpmatch} offset the four free parameters
of the slow and fast subsystems so that the resulting set of solutions
is typically one-parametric, \ie\ it is  curve in the parameter
space. The projection of this curve on the $(\c,\BCL)$ is the
CV restitution curve. Given appropriate analytical approximation of
the relevant dependencies, which could be obtained for instance by
asymptotic means or by fitting, solving the resulting finite
(transcendental) system will produce analytical approximation for the
restitution curve. Doing so for Beeler-Reuter model is however beyond 
the scope of this paper and we restrict to demonstrate the validity of
our asymptotic approach for this model by solving the asymptotic
matching problem numerically.

In solving the problem  \eq{BRbvpslow}, \eq{BRbvpfast} and
\eq{BRbvpmatch} numerically, the following features need to be taken
into account
\begin{enumerate}[(a)]
\item The fast-time problem is posed on an infinite interval.
\item At the same time the slow-time system is posed on a finite
  interval.
\item The length of that finite interval
  is the wave length of the periodic travelling wave, \ie\ it is an
    unknown variable.
\item The fast-time system has piece-wise right-hand sides.
\item The pinning condition needs to be imposed in the fast subsystem.
  Since the fast system is piece-wise, it is convenient to impose it
  at a boundary between pieces.
\item The slow variable $\j$ appears as a parameter in the fast-time
  system.
\end{enumerate}

These features can be addressed by a number of
well-known techniques and we refer the interested reader to
\cite{Ascher-Russell-1981} for a general discussion and to
\cite{Simitev-Biktashev-2006} for a numerical implementation of a
similar problem. In short, the issue of the boundary
conditions at infinity can be resolved by considering a finite
interval with boundary conditions obtained as a solution of the problem
linearised about the asymptotic equilibria. This finite interval is
then dissected into three subintervals to take care of the piece-wise
definition of the equations. The three subinterlals together with the
interval on which the slow-time system is posed are then mapped to the
interval $[0,1]$ by introducing appropriate scaling factors. The pinning
condition can be easily incorporated at one of the internal matching
points. Finally, in this representation equations \eq{BRbvpfast}--\eq{BRbvpmatch}
can be solved
by any standard boundary value problem solver such as Maple's
{\tt dsolve} \cite{Maple}, NAG's {\tt d02raf} \cite{NAG} and others.

\dblfigure{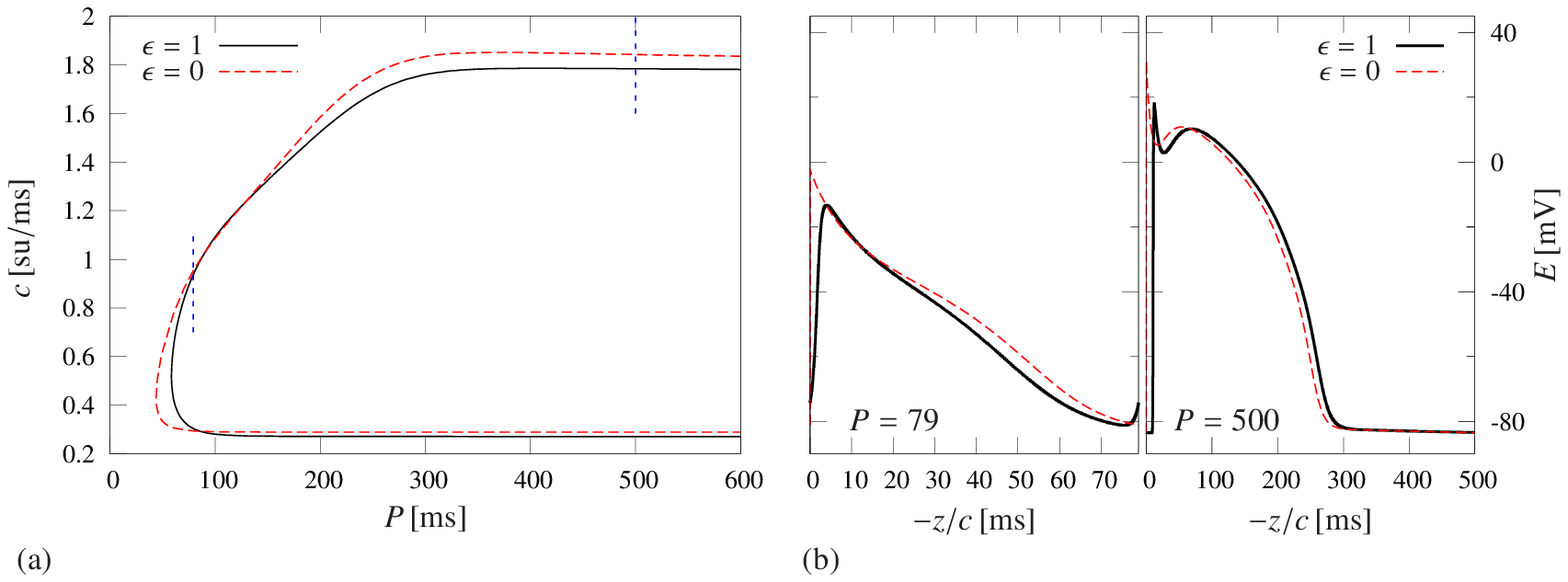}{
  (color online) 
  (a) The numerical CV restitution curves of the full non-asymptotic problem
  \eqref{eq:full} for the original Beeler-Reuter model \cite{Beeler-Reuter-1977}
  (solid line, red online) compared to the  numerical CV restitution curve of the
  asymptotic  problem \eq{BRbvpslow}, \eq{BRbvpfast} and
  \eq{BRbvpmatch} (dashed line, black online).  (b) The AP profiles
  corresponding to $\BCL=79\,\ms$ and $\BCL=500\,\ms$ (same line types). 
}{brsol}

The resulting asymptotic CV restitution curve is shown in \fig{brsol}
by a dashed thin red line. The bold solid black line in the figure
corresponds to the CV restitution curve found from the full
non-asymptotic boundary value problem \eq{full} written for the
full Beeler-Reuter model \eq{BR} at $\eps=1$.  The two curves
demonstrate a good quantitative agreement.

We would like to emphasize here that, while it was still possible to
solve the full problem numerically in this case, and the asymptotic
problem was solved numerically too, solution of the full problem was a
substantially more difficult task than the computation of the
asymptotic CV restitution curve. The problem is certainly well-posed
but it is very stiff and it required a prolonged experimentation with
a variety of software tools and parameter continuation techniques. It
is also worth recalling that the Beeler-Reuter model is not as
complicated as contemporary models are, which leads us to expect that
the non-asymptotic problem for such models is even more difficult to
solve.

\section{Discussion}
\label{discussion}

\dblfigure{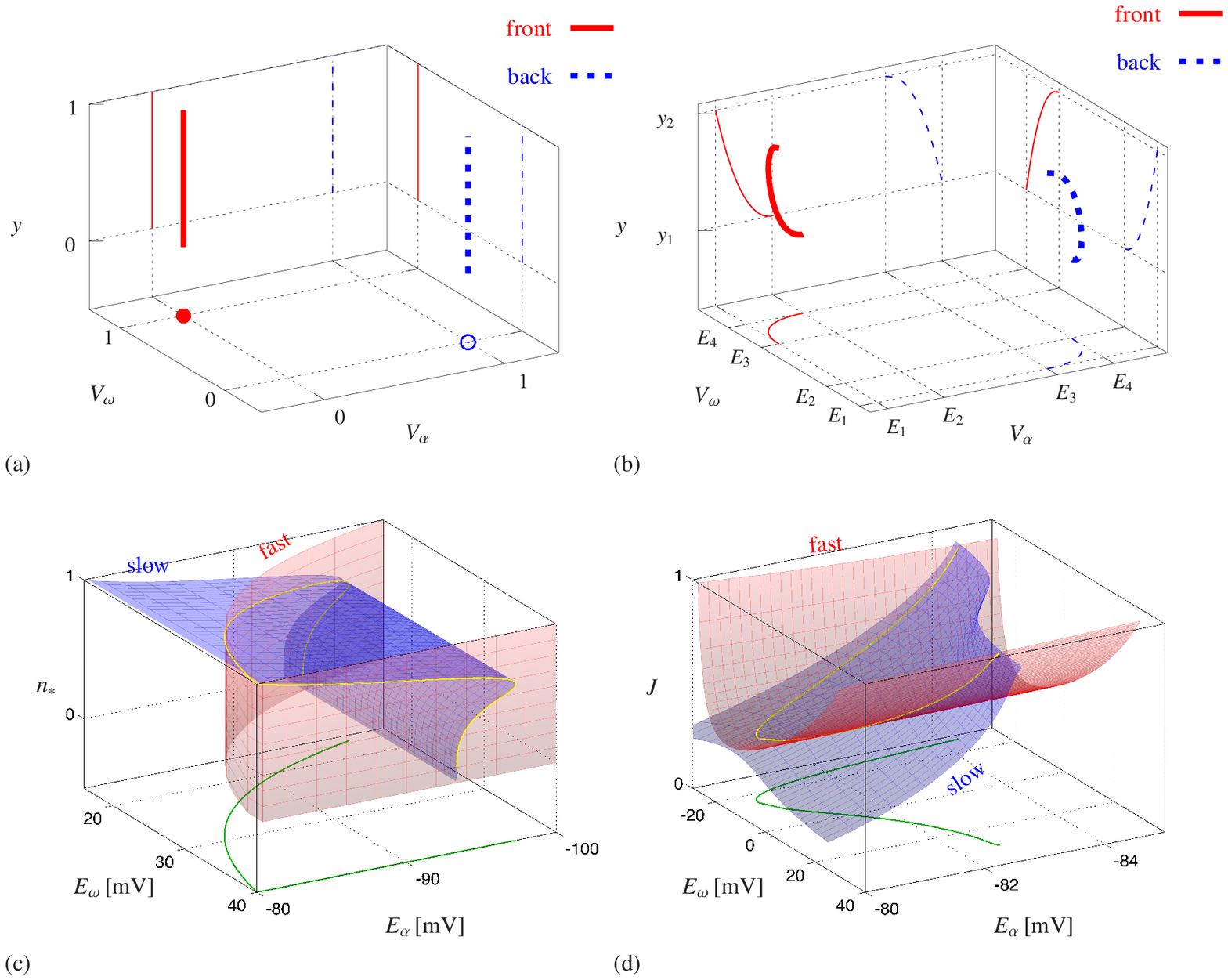}{
  (color online) 
  Pre/post-front voltage selection in the four different models, (a)
  Barkley, (b) FitzHugh-Nagumo, (c) Caricature Noble, (d)
  Beeler-Reuter.
  (a,b) Relationship between $\ualpha$, $\uomega$ and corresponding
  value of the slow variable $\v$ in bold lines, with their
  projections on the coordinate planes in thin lines. Each of the cases
  has two dependencies, for the front of a pulse and for the back of a
  pulse.
  (c,d) Relationship between $\Ealpha$, $\Eomega$ and an indicated
  slow variable, as semi-transparent surfaces, as following from the
  relevant fast and slow subsystems. The lines of intersection of
  the surfaces and their projection onto $(\Ealpha,\Eomega)$ plane are
  also shown.
}{ao}

\paragraph{Summary}
We have demonstrated that singular perturbation theory based on the
largeness of the maximal value of the sodium current $\INa$
compared to other currents and quality of the $\INa$ ionic gates
(smallness of $\mbar(\E)$ and $\hbar(\E)$ in certain voltage ranges),
is capable of reproducing essential spatiotemporal phenomena, using
conduction velocity restitution curve as the simplest nontrivial
example involving both the fast scale and slow scale. We have
explicitly compared the mathematical technique involved here, with
similar problems in the 
classical FitzHugh-Nagumo (FHN) like models of excitable media.
Apart from the different number of equations and the more
complicated right-hand sides, we have identified in the cardiac
models qualitatively new features of topological nature.

\paragraph{Classical simplified excitable models vs ionic cardiac models}
\Fig{ao} illustrates the fundamental difference between FHN-like and
cardiac models as far as the problem of $\Ealpha$ and $\Eomega$ selection
is concerned.  In FHN-like systems (panels (a) and (b)), if the values
of the slow variables are given, then $\Ealpha$ and $\Eomega$ are
uniquely defined, up to the choice between front (jump up) and back
(jump down). So, for one slow variable $\v$ as in the examples 
considered in section \ref{Tikhonov}, there are one-dimensional
manifolds representing all possible fronts and backs. In contrast
in cardiac equations (panels (c) and (d), we have two-dimensional
manifolds in place of one-dimensional 
manifolds.  This is related to the fact that in the parametric
embedding we use, the transmembrane voltage, like a double-faced
Janus, is both a ``slow'' and a ``fast'' variable. In its slow
capacity, it contributes to the pre-front conditions, among other slow
variables. In its fast capacity, it participates in the excitation
front. So given one ``truly slow'' variable ($\nstar$ in panel (c) and
$\jfast$ in panel (d)), we have not a one-dimensional, but a
two-dimensional manifold representing possible fronts, as even when
the value of that slow variable is fixed, the prefront voltage
$\Ealpha$ and, consequently, post-front voltage $\Eomega$ are still
undefined.  Hence to obtain the restitution curve, which is a
1-dimensional manifold, we have to find an intersection of the
manifold representing possible fast fronts with the manifold
representing possible slow solutions, which is something that we don't
do for the FHN-like systems.

Another qualitative difference is, of course, the number of fast
solutions per one excitation pulse. In FHN-type systems, it is
essential that there is a back corresponding to every front, as the
systolic and diastolic pieces of the reduced slow manifold are separated from
each other. In our asymptotic embedding of cardiac models, the
systolic and diastolic pieces of the slow set (which is now not a
manifold) are connected via a piece where both $\n$ and $\m$ gates are
firmly closed and $\E$ is in its slow-variable capacity. Here it should
be noted that existence of a back is not a necessary feature of a
FHN-type system, as in presence of two or more slow variables it is
possible to have a slow manifold with a cusp singularity so that its
systolic and diastolic pieces of the are connected via a monostable
pieces, as in the example suggested by
Zeeman~\cite{Zeeman-1972}. However, although the structure of ionic
models admits, in principle, such
manifolds~\cite{Suckley-Biktashev-2003}, it has not been identified in
any cardiac models so far.

\paragraph{The numerical method for  dynamic CV restitution curves}
In this work, we have proposed a computationally efficient method of
calculating an ideal case of the so called ``dynamic'' or
``steady-state'' restitution protocol, with exactly periodic
propagating pulses. The method exploits asymptotic splitting the
problem into two parts, the slow and the fast subsystems. The
advantage of such a split is that each of the subsystems no longer
depends on the small parameter due to the largeness of $\INa$
and quality of its ionic gates, so they are significantly less
stiff than the original full problem, hence the computational efficiency. The
well-posedness of the problems arising from the asymptotic splitting
is not obvious \apriori, since the asymptotic embedding we use is
non-Tikhonov, and the general results from
singular perturbation theory are not applicable to our case. We have
thus taken care to prove the well-posedness, at least for the case
considered.  

\paragraph{Perspectives: more complicated spatiotemporal regimes}
There are several other protocols used in defining restitution curves,
see \eg\ \cite{Schaeffer-etal-2007} and references therein. They do
not correspond to periodic steadily propagating wave
solutions. Nevertheless, we expect that the proposed asymptotic
splitting should still work there, under some natural assumptions.

So, if propagation of waves is not too complicated, then evolution
of the system away from the fronts can be expected to be well
approximated by the corresponding slow subsystem, which is a non-stiff ordinary
differential equation for every point of the medium. After
application of the propagating wave ansatz, the slow subsystems of the
Caricature Noble and Beeler-Reuter become systems~\eq{caric-slow}
and \eq{BRbvpslow} respectively. Here we note, that without such
ansatz, the slow subsystems would be differential equations with
time as the independent variable, with exactly the same structure as
\eq{caric-slow} and \eq{BRbvpslow} only with a different scaling, so
all the above analysis applies. Fronts themselves for most part of their
evolution can be considered locally as steadily propagating nearly
plane waves, well described by systems like~\eq{caric-fast} and 
\eq{BRbvpfast}. 

Under these assumptions, according to the analysis of the fast
subsystem, the conduction velocity and post-front voltage can be
found as functions of the pre-front voltage and, if relevant,
pre-front value of gate $\j$, and the dynamics of the front is
governed by an ordinary (in the case of one spatial dimension)
differential equation for the position of the front as a function of
time. This gives an unusually coupled system of ordinary
differential equations: the local dynamics provide right-hand sides
for the equation of motion of the next front, and trajectory of the
front provides initial condition for the subsequent piece of slow
dynamics.  This unusual ODE system can predict non-stationary
evolution of the excitation patterns, including restitution
protocols. In different settings, this approach has been used \eg\
in \cite{Biktashev-Tsyganov-2005} and
\cite{Echebarria-Karma-2007}. As demonstrated
in~\cite{Echebarria-Karma-2007}, such unusually coupled ODE system
can form a basis for investigation of such complicated spatiotemporal
phenomenon as the spatiotemporal dynamics of cardiac
alternans.~\footnote{ %
  Echebarria and Karma~\cite{Echebarria-Karma-2007} considered a simplified
  model with some features of the ionic models and similar to the Caricature
  Noble model considered there, but only two variables, without 
  any slow gates.
} %
So extension of the present results to other, more
complicated and important spatiotemporal regimes, seems to be a
natural and imminent next step. %

\paragraph{Perspective: the problem of restitution memory}
The above considerations lead us to the problem of
rate-dependence of restitution curves, or the so called ``memory''
effects, see \eg~\cite{Schaeffer-etal-2007} and references therein. A
typical approach to studying memory effects is purely phenomenological:
memory variables in the restitution relationships are usually
postulated and their properties are obtained inductively from
measurements of the differences between results from various
restitution protocols. Asymptotic analysis such as the one used in the
present work offers a deductive \textit{ab initio} way of treatment
of the memory variables. In such setting, the number of memory variables
equals dimensionality of the phase space of the slow subsystem minus
one, and the memory variables themselves may be identified with values
of the slow variables apart from $\E$, measured during an excitation
front. For example, in the Caricature Noble model here there would be
one memory variable, say taken to be $\nstar$, and in the
Beeler-Reuter model, there would be $\dim{\y}=5$ of them, and they may
be identified with the values of the slow variables $\j$, $\xone$
$\dgate$, $\fgate$ and $\Ca$ as measured at the front, thus making the
restitution much more difficult to predict. However, empirical
evidence from simulations of modern detailed ionic models suggests
that variety of slow trajectories may be much narrower, \ie\ they may
\defacto\ be restricted to a manifold of a smaller
dimensionality~\cite{Simitev-Biktashev-2006}.  A possible theoretical
explanation for such dimensionality reduction is the presence of
further small parameters within the slow subsystem; this mechanism was
considered in detail in \cite{Biktashev-etal-2008}. So detailed
studies of finer asymptotic structure of slow subsystems of
practically interesting ionic models are a promising direction for
further studies, which may help in understanding the memory effects in
restitution. It is well known that memory effects can play
considerable part in development of alternans and therefore in
development of cardiac arrhythmias (see~\eg~\cite{Mironov-etal-2008}
for a recent study) and are for this reason of considerable interest.

\paragraph{Perspective: a novel numerical method for excitation propagation}
The above considerations outline possible \emph{qualitative}
applications of the present study.  A possible \emph{quantitative}
application is an advanced method for calculation of activation
sequences, which can be achieved by the aforementioned coupling of differential
equations for the local slow dynamics and for the front motion. For
these applications, two or three spatial dimensions rather than one
are interesting, hence the equation of motion for the front is a
partial differential equation of motion of a line (in 2D) or surface
(in 3D). One immediate difference is that propagation of the front
shall no longer depend only on the prefront voltage and $\j$-gate, but
also on the spatial configuration of the front. It is well known that
unless the shape of the front deviates very strongly from plain, the
effect of its shape is mostly accounted for through its mean
curvature, and we have demonstrated that the effect of curvature can
be easily incorporated into the asymptotic description of the front
dynamics~\cite{Simitev-Biktashev-2006}. This approach can be used to
describe normal activation sequences in the heart, when the graph of
the front solution in the space-time is a manifold without internal
boundaries. More serious problems occur if there are propagation
blocks and/or wave breaks, which introduce boundaries of the front
manifold in space-time. In such cases, a separate asymptotic description
for the codimension-two areas, the wave break trajectories and the
propagation block loci, are needed; obtaining such asymptotic
description is another important direction for further research. 

An obvious caveat here is, as with any asymptotic approach, that the
asymptotics are valid in the limit of $\eps\to0$ whereas we apply it
for a finite value of that small parameter. Hence applicability of
this approach on the quantitative level will depend on whether the
error generated by such approximation is tolerable for the particular
application; however, it should be remembered that higher-order
terms if necessary can improve the accuracy, see an example in
\cite{Biktashev-etal-2008}.


\end{document}